\documentclass[11pt]{article}
\usepackage[numbers]{natbib}

\usepackage{comment}
\usepackage{ifthen}

\usepackage{amsmath,amssymb,amsfonts,amsthm}
\usepackage{graphicx}
\usepackage{bm}
\usepackage{multirow}
\usepackage{setspace}
\usepackage[margin=1in]{geometry}
\usepackage[dvipsnames,usenames]{xcolor}
\usepackage[letterpaper=true,colorlinks=true,pdfpagemode=none,urlcolor=RoyalBlue,linkcolor=RoyalBlue,citecolor=OliveGreen,pdfstartview=FitH]{hyperref}
\usepackage{prettyref}

\let\pref=\prettyref
\newcommand{\savehyperref}[2]{\texorpdfstring{\hyperref[#1]{#2}}{#2}}

\newrefformat{eq}{\savehyperref{#1}{\textup{(\ref*{#1})}}}
\newrefformat{chp}{\savehyperref{#1}{Chapter~\ref*{#1}}}
\newrefformat{lem}{\savehyperref{#1}{Lemma~\ref*{#1}}}
\newrefformat{def}{\savehyperref{#1}{Definition~\ref*{#1}}}
\newrefformat{thm}{\savehyperref{#1}{Theorem~\ref*{#1}}}
\newrefformat{cor}{\savehyperref{#1}{Corollary~\ref*{#1}}}
\newrefformat{cha}{\savehyperref{#1}{Chapter~\ref*{#1}}}
\newrefformat{sec}{\savehyperref{#1}{Section~\ref*{#1}}}
\newrefformat{app}{\savehyperref{#1}{Appendix~\ref*{#1}}}
\newrefformat{tab}{\savehyperref{#1}{Table~\ref*{#1}}}
\newrefformat{fig}{\savehyperref{#1}{Figure~\ref*{#1}}}
\newrefformat{hyp}{\savehyperref{#1}{Hypothesis~\ref*{#1}}}
\newrefformat{alg}{\savehyperref{#1}{Algorithm~\ref*{#1}}}
\newrefformat{item}{\savehyperref{#1}{Item~\ref*{#1}}}
\newrefformat{step}{\savehyperref{#1}{step~\ref*{#1}}}
\newrefformat{conj}{\savehyperref{#1}{Conjecture~\ref*{#1}}}
\newrefformat{fact}{\savehyperref{#1}{Fact~\ref*{#1}}}
\newrefformat{prop}{\savehyperref{#1}{Proposition~\ref*{#1}}}
\newrefformat{claim}{\savehyperref{#1}{Claim~\ref*{#1}}}

\newenvironment{prevproof}[2]{\noindent {\em {Proof of
{#1}~\ref{#2}:}}}{$\blacksquare$\vskip \belowdisplayskip}

\theoremstyle{definition}
\newtheorem{definition}{Definition}[section]

\theoremstyle{theorem}
\newtheorem{theorem}{Theorem}[section]
\newtheorem{lemma}[theorem]{Lemma}

\DeclareMathOperator*{\argmax}{\arg\max}

\newcommand{\mhr}{{\textit{\rm exp}}}

\renewcommand{\Pr}{{\textit{\bf Pr}}}
\newcommand{\E}{{\textit{\bf E}}}
\newcommand{\kl}{{D_{\textit{KL}}}}
\newcommand{\D}{{\mathcal{D}}}

\newcommand{\prob}[2][]{\text{\bf Pr}\ifthenelse{\not\equal{}{#1}}{_{#1}}{}\!\left[#2\right]}
\newcommand{\expect}[2][]{\text{\bf E}\ifthenelse{\not\equal{}{#1}}{_{#1}}{}\!\left[#2\right]}

\newcommand{\eps}{\epsilon}

\newcommand{\val}{v}

\newcommand{\samples}{\val_1,\ldots,\val_m}

\newcommand{\dist}{D}

\newcommand{\dens}{f}

\begin{document}

%\markboth{Huang, Mansour, and Roughgarden}{Making the Most of Your Samples}

\title{Making the Most of Your Samples}
\author{
Zhiyi Huang\thanks{University of Hong Kong. This research was done while the author was a postdoc at Stanford University.
Email: \href{mailto:zhiyi@cs.hku.hk}{zhiyi@cs.hku.hk}.}
\and
Yishay Mansour\thanks{Microsoft Research and Tel Aviv University. This research was supported in part by The Israeli Centers of Research Excellence (I-CORE) program, (Center  No. 4/11), by a grant from the Israel Science Foundation (ISF), by a grant from United States-Israel Binational Science Foundation (BSF), and by a grant from the Israeli Ministry of Science (MoS).
Email: \href{mailto:mansour@tau.ac.il}{mansour@tau.ac.il}.}
\and
Tim Roughgarden\thanks{Stanford University. This research was supported in part by NSF grants CCF-1016885 and CCF-1215965, and an ONR PECASE Award.
Email: \href{mailto:tim@cs.stanford.edu}{tim@cs.stanford.edu}.}
}

%\onehalfspacing

\begin{titlepage}
\thispagestyle{empty}

\date{}

\maketitle

\begin{abstract}
\thispagestyle{empty}
%We consider selling a good to a single buyer with a valuation drawn
%from an unknown distribution $\dist$, and study the best-possible
%approximation of the expected revenue as a function of the
%available data, in the form of i.i.d.\ samples from $\dist$.  For
%example, the Bulow-Klemperer theorem shows how to obtain a
%$\tfrac{1}{2}$-approximation with just one sample (by posting a price
%equal to the sample), assuming only that $\dist$ is a regular
%distribution.  What is the optimal way to use one or multiple samples?
We study the problem of setting a price for a
potential buyer with a valuation drawn from an unknown
distribution~$\dist$.
The seller has ``data'' about $\dist$ in the form of $m \ge
1$ i.i.d.\ samples, and the algorithmic challenge is to use these
samples to obtain expected revenue as close as possible to what could
be achieved with advance knowledge of $\dist$.

Our first set of results quantifies the number of samples $m$ that are
necessary and sufficient to obtain a $(1-\eps)$-approximation.
For example, for
an unknown distribution that satisfies the monotone hazard rate (MHR)
condition, we prove that $\tilde{\Theta}(\eps^{-3/2})$ samples are
necessary and sufficient.
Remarkably, this is fewer samples than is necessary to accurately
estimate the expected revenue obtained for such a distribution
by even a single reserve price.
We also prove essentially tight sample complexity bounds
for regular
distributions, bounded-support distributions, and a wide class of
irregular distributions.
Our lower bound approach, which
applies to all randomized pricing strategies,
borrows tools from differential privacy and
information theory, and we believe it
could find further applications in auction theory.

%For regular distributions, where previous work shows that $O(\eps^{-3}
%\log \eps^{-1})$ samples are sufficient for a
%$(1-\eps)$-approximation, we provide a nearly matching lower bound of
%$\Omega(\eps^{-3})$.

Our second set of results considers the single-sample case.
While no deterministic pricing strategy is better
than $\tfrac{1}{2}$-approximate for regular distributions,
for MHR distributions we show how to do better:
there is a simple deterministic pricing strategy
that guarantees expected revenue at least $0.589$ times the
maximum possible.  We also prove that no
deterministic pricing strategy achieves an
approximation guarantee better than $\frac{e}{4} \approx .68$.
%Our positive results for these single-agent problems directly lead to
%improved multi-agent prior-independent mechanisms.
\end{abstract}

%\footnotetext[1]{
%This research was done while the author was a postdoc at Stanford University.
%Email: \href{mailto:zhiyi@cs.hku.hk}{zhiyi@cs.hku.hk}.
%}

%\footnotetext[2]{
%This research was supported in part by The Israeli Centers of Research Excellence (I-CORE) program, (Center  No. 4/11), by a grant from the Israel Science Foundation (ISF), by a grant from United States-Israel Binational Science Foundation (BSF), and by a grant from the Israeli Ministry of Science (MoS).
%Email: \href{mailto:mansour@tau.ac.il}{mansour@tau.ac.il}.
%}

%\footnotetext[3]{
%This research was supported in part by NSF grants CCF-1016885 and CCF-1215965, and an ONR PECASE Award.
%Email: \href{mailto:tim@cs.stanford.edu}{tim@cs.stanford.edu}.
%}

%\setcounter{footnote}{3}
\end{titlepage}

\section{Introduction}
\label{sec:intro}

We study the basic pricing problem of making an ``optimal''
take-it-or-leave-it price to a potential buyer with an unknown
willingness-to-pay (a.k.a.\ {\em valuation}).
Offering a price of~$p$ to a buyer with valuation~$v$ yields revenue
$p$ if $v \ge p$, and 0 otherwise.
The traditional approach in theoretical computer science to such
problems is to assume as little as possible about the buyer's
valuation --- for example, only lower and upper bounds on its value
--- and to compare the performance of different prices using
worst-case analysis.
The traditional approach in economics is to assume that the buyer's
valuation is drawn from a distribution $\dist$ that is known to the
seller, and to use average-case analysis.
In the latter case, the optimal solution is clear --- it is the
{\em monopoly price} $\max_{p \ge 0} p \cdot (1-F(p))$, where $F$ is the CDF of $\dist$.

\citet{CR14} recently proposed adapting the formalism of learning
theory~\cite{V84} to interpolate between the traditional worst- and
average-case approaches, in the context of single-item auction design.
The idea is to parameterize a seller's knowledge about an unknown
distribution~$\dist$ through a number~$m$ of i.i.d.\ samples from
$\dist$.
When $m=0$ this is equivalent to the worst-case approach, and as $m
\rightarrow \infty$ it becomes equivalent to the average-case
approach.
The benchmark is the maximum expected revenue obtainable when the
distribution $\dist$ is known a priori.
The algorithmic challenge is to use the $m$
samples from $\dist$ to get expected revenue as close to this
benchmark as possible, no matter what the underlying distribution
$\dist$ is.\footnote{There are, of course, other ways one can parameterize partial knowledge about valuations.  
See e.g.~\citet{ADMW13,CMZ12} for alternative approaches.}

This ``hybrid'' model offers several benefits.
First, it is a relatively faithful model of many realistic computer
science applications, where data from the past is assumed to be a
reasonable proxy for future inputs and guides the choice of an
algorithm.
For example, in Yahoo!'s keyword auctions, the choice of reserve
prices is guided by past bid data in a natural way~\citep{OS11}.
Second, the model is a potential ``sweet spot'' between worst-case and
average-case analysis, inheriting much of the robustness of the
worst-case model (since we demand guarantees for every underlying
$\dist$) while allowing very good approximation guarantees with
respect to a strong benchmark.
Third, by analyzing the trade-offs between the number of samples $m$
available from $\dist$ and the best-possible worst-case approximation
guarantee, the analysis framework implicitly
quantifies the value of data (i.e., of additional samples).
It becomes possible, for example, to make statements like ``4 times as
much data improves our revenue guarantee from 80\% to 90\%.''
Finally, proving positive results in this model involves rigorously
justifying natural methods of incorporating past data
into an algorithm, and this task is interesting in its own right.

%The classical theory of revenue-optimal auctions measures auction
%revenue using Bayesian, or average-case, analysis.
%Bidders' valuations are assumed to be drawn from a distribution, and
%the optimal auction is defined as the one with the highest expected
%revenue.
%The optimal auction depends on the assumed distribution, in some cases
%in a detailed way.
%In most applications, and especially in computer science contexts,
%this distribution is derived from past data.

%The detailed dependence of the optimal auction on the input
%distribution is problematic for at least two reasons.
%First, a seller may have only a rough estimate of the input
%distribution (e.g., from a modest amount of past data), and it is not
%clear that an optimal auction for an approximation of the distribution
%yields a near-optimal auction for the true distribution (e.g., due to
%overfitting the data).
%Second, even if the underlying input distribution is fully known, the
%theoretically optimal auction can be too complex for practical use.
%These two critiques of traditional optimal auctions have motivated
%much recent work on the design and analysis of simple near-optimal
%auctions that have minimal or no dependence on the underlying
%valuation distributions (e.g., see \citet{Hartline}).

\subsection{Our Results}

Formally, we study a single seller of some good, and a single buyer
with a private valuation $\val$ for the good drawn from an unknown
distribution $\dist$.
The seller has access to $m \ge 1$ i.i.d.\ samples $\samples$ from $\dist$.
The goal is to identify,
among all {\em $m$-pricing strategies} --- functions
from a sample $\samples$ to a price $p$ --- the strategy that has the
highest
expected revenue.  The expectation here is over $m+1$ i.i.d.\ draws
from $\dist$ --- the samples $\samples$ and the unknown valuation
$\val$ of the buyer --- and the randomness of the pricing strategy.
The {\em approximation guarantee} of a pricing strategy $p(\cdot)$ for a
set $\D$ of distributions is its worst-case (over $\D$) approximation
of the (optimal) expected revenue obtained by the monopoly price:
% (which is simply $\max_p p(1-\dist(p))$
%for a distribution $\dist$):
\[\inf_{\dist \in \D} \frac{\expect[\samples \sim \dist]{p(\samples) \cdot (1-F(p(\samples)))}}{\max_p p(1-F(p))} \enspace, \]
where $F$ is the CDF of $\dist$.

We first describe our results that quantify the inherent trade-off
between the number of samples~$m$ and the best-possible approximation
guarantee of an $m$-pricing strategy; see also \pref{tab:many}.
%The primary goal of this paper is to understand the optimal
%approximation of the expected revenue possible as a function of the
%available data, in the form of i.i.d.\ samples from an unknown
%valuation distribution.
Some restriction on the class $\D$ of allowable distributions is
  necessary for the existence of pricing strategies with any
  non-trivial approximation of the optimal expected revenue.\footnote{The family of distributions that take on a
  value $M$ with probability $\tfrac{1}{M}$ and 0 with probability
  $1-\tfrac{1}{M}$ highlight the difficulty.}
We give essentially tight bounds on the number of samples that are
necessary and sufficient to achieve a target approximation of
$1-\eps$, for all of the choices of the class~$\D$ that are common in
auction theory.
For example, when $\D$ is the set of distributions that satisfy the
monotone hazard rate (MHR) condition\footnote{$\dist$ satisfies the
  monotone hazard condition if
  $\tfrac{\dens(\val)}{1-F(\val)}$ is nondecreasing; see
  \pref{sec:prelim} for details.}, we prove
that $m = \Omega(\eps^{-3/2})$ samples are necessary and
that $m = O(\eps^{-3/2} \log \eps^{-1})$ samples are sufficient
to achieve an approximation guarantee of $1-\eps$.\footnote{We
  suppress only universal constant factors, which do not depend
  on the specific distribution $\dist \in \D$.  Such uniform
  sample complexity bounds are desirable because the valuation distribution
  is unknown.  Law of Large Numbers-type arguments do not generally
  give   uniform bounds.}
This bound holds more generally for the class of ``$\alpha$-strongly
regular distributions'' introduced in~\citet{CR14} (for fixed $\alpha >
0$).
When $\D$ is the (larger) set of regular
distributions\footnote{$\dist$ is regular if the ``virtual
  valuation function'' $\val - \tfrac{1-F(v)}{f(v)}$ is
  non-decreasing; see \pref{sec:prelim} for details.},
we prove that the sample complexity is
$\tilde{\Theta}(\eps^{-3})$.
When $\D$ is the set of arbitrary distributions with support contained
in $[1,H]$, the sample complexity is $\tilde{\Theta}(H\eps^{-2})$.
We also give essentially optimal sample complexity bounds for
distributions that are parameterized by the probability of a sale at
the monopoly price; see \pref{sec:highlights} for more
discussion.
On the upper bound side, our primary contribution is the
bound for MHR and strongly regular distributions.\footnote{The upper bound for
  regular distributions was proved in~\citet{DRY10} and the upper bound
  for bounded valuations can be deduced from~\citet{BBHM05}.}
All of our lower bounds, which are information-theoretic and apply to
arbitrary randomized pricing strategies, are new.

%%We next outline our results; see also \pref{tab:many} and
%%\pref{tab:onesample}.
%Our first set of results quantifies the number of samples $m$ that are
%necessary and sufficient for the existence of a $(1-\eps)$-approximate
%$m$-pricing strategy, as a function of $\eps$.
%We give upper and lower bounds that match up to a logarithmic factor,
%for both the sets of regular and monotone hazard rate
%(MHR)\footnote{$F$ satisfies the monotone hazard condition if
%  $\tfrac{\dens(\val)}{1-\dist(\val)}$ is nondecreasing; see
%  \pref{sec:prelim} for details.} distributions.
%Previous results of \citet{DRY10} can be interpreted as
%$(1-\eps)$-approximate pricing strategies provided $m =
%\Omega(\eps^{-2} \log \eps^{-1})$ when $\D$ is the set
%of MHR distributions and provided
%$m = \Omega(\eps^{-3} \log \eps^{-1})$ when $\D$ is the
%set of regular distributions.
%For the former class, we provide a new upper bound --- there is a
%$(1-\eps)$-approximate pricing strategy provided $m =
%\Omega(\eps^{-3/2} \log \eps^{-1})$.\footnote{This
%  bound, like those in \citet{DRY10},
%is uniform w.r.t.\ the set $\D$ of allowable distributions and
%therefore does not obviously follow from the Law of Large Numbers or
%standard concentration bounds.}

\begin{table}
\renewcommand{\arraystretch}{1.2}
%\tbl{}{
\centering
\begin{tabular}{|c|l@{~}l|l@{~}l|}
\hline
& \multicolumn{2}{c|}{Upper Bound} & \multicolumn{2}{c|}{Lower Bound} \\
\hline
MHR & $O(\epsilon^{-3/2} \log \epsilon^{-1})$ & (\savehyperref{thm:mhrmanyub}{Thm.~\ref*{thm:mhrmanyub}}) & $\Omega(\epsilon^{-3/2})$ & (\savehyperref{thm:mhrmanylb}{Thm.~\ref*{thm:mhrmanylb}}) \\
Regular & $O(\epsilon^{-3} \log \epsilon^{-1})$ & (\cite{DRY10}) & $\Omega(\epsilon^{-3})$ & (\savehyperref{thm:regmanylb}{Thm.~\ref*{thm:regmanylb}}) \\
General & $O(\delta^{-1} \epsilon^{-2} \log (\delta^{-1} \epsilon^{-1}) )$ & (\savehyperref{thm:generalmanyub}{Thm.~\ref*{thm:generalmanyub}}) & $\Omega(\delta^{-1} \epsilon^{-2})$ & (\savehyperref{thm:generalmanylb}{Thm.~\ref*{thm:generalmanylb}}) \\
Bounded Support & $O(H \epsilon^{-2} \log (H \epsilon^{-1}) )$ & (\savehyperref{thm:generalmanyub2}{Thm.~\ref*{thm:generalmanyub2}} and \cite{BBHM05}) & $\Omega(H \epsilon^{-2})$ & (\savehyperref{thm:generalmanylb2}{Thm.~\ref*{thm:generalmanylb2}}) \\
\hline
\end{tabular}%}
\caption{Sample complexity of a $(1-\epsilon)$-approximation.
For bounded-support distributions, the support is a subset of $[1,H]$.
For general distributions, the benchmark is the optimal revenue of prices with sale probability at least $\delta$.}
\label{tab:many}
\end{table}

Our second set of results considers the regime where the
seller has only one sample ($m=1$) and wants to use it in the optimal
deterministic way.\footnote{We offer the problems of determining the
  best randomized pricing strategy for $m=1$ and the best way to use a
  small $m \ge 2$ number of samples as challenging and exciting
  directions for future work.}
\citet{DRY10} observed that an elegant result from auction theory, the
Bulow-Klemperer theorem on auctions vs.\ negotiations~\cite{BK96},
implies that the 1-sample pricing strategy $p(\val) = \val$ has an
approximation guarantee of $\tfrac{1}{2}$ when $\D$ is the set of
regular distributions.\footnote{\citet{DRY10} observed this in the
  context of the design and analysis of prior-independent auctions.
Plugging our better bounds for single-sample pricing starategies with MHR
distributions into the framework of~\citet{DRY10} immediately yields
analogously better prior-independent mechanisms.}
It is not hard to prove that there is no better deterministic pricing
strategy for this set of 
distributions.  We show how to do better, however, when $\D$ is the
smaller set of MHR distributions:
%For the class of regular distributions, we prove that the
%Bulow-Klemperer pricing strategy $p(\val) = \val$ is optimal: no
%$1$-pricing strategy is better than $\tfrac{1}{2}$-approximate.
%For MHR distributions, we show that the strategy $p(\val) = \val$ is
%still no better than $\tfrac{1}{2}$-approximate and that other
%strategies can be strictly better.  We prove that
a simple $1$-pricing strategy of
the form $p(\val) = c\val$ for some $c < 1$ has an approximation
guarantee of  $0.589$.
We also prove that no deterministic
$1$-pricing strategy is better than an
$\tfrac{e}{4}$-approximation for MHR distributions, and that no
continuously differentiable such strategy is better than a
$0.677$-approximation.

%The Bulow-Klemperer theorem~\cite{BK96} gives a sense in which extra
%competition is better than distributional information:
%for a single-item auction with bidders' valuations drawn i.i.d.\ from
%a regular
%distribution $\dist$\footnote{$\dist$ is regular if the ``virtual
%  valuation function'' $\val - \tfrac{1-\dist(v)}{f(v)}$ is
%  non-decreasing; see \pref{sec:prelim} for details.}, the expected
%revenue of an optimal auction with $n \ge 1$ bidders is at most that
%of a Vickrey auction with $n+1$ bidders.
%\citet{DRY10} observed a corollary of the $n=1$ special case of this
%result:
%for a single bidder with valuation drawn from $\dist$, the expected
%revenue obtained from a random reserve price (drawn from $\dist$) is
%at least 50\% times that from an optimal reserve price.
%In the context of the present work, this observation implies that
%there is already a pricing strategy with a non-trivial guarantee when
%$m = 1$ and $\D$ is the set of regular distributions:
%the strategy $p(\val) = \val$ is a $\tfrac{1}{2}$-approximation.

\begin{table}
\renewcommand{\arraystretch}{1.2}
\centering
%\medskip
%\tbl{Optimal approximation ratio with a single sample.}{
\begin{tabular}{|c|l@{~}l|l@{~}l|}
\hline
& \multicolumn{2}{c|}{Positive Result} & \multicolumn{2}{c|}{Negative Result} \\
\hline
Regular & $\ge 0.5$ & (\cite{BK96, DRY10}) & $\le 0.5$ & (\savehyperref{thm:regonesamplelb}{Thm.~\ref*{thm:regonesamplelb}}) \\
MHR & $\ge 0.589$ & (\savehyperref{thm:mhrsingleub}{Thm.~\ref*{thm:mhrsingleub}}) & $\le 0.68$ & (\savehyperref{thm:mhronesamplelb2}{Thm.~\ref*{thm:mhronesamplelb2}}) \\
\hline
\end{tabular}%}
\caption{Optimal approximation ratio with a single sample.}
\label{tab:onesample}
\end{table}

\subsection{A Few Technical Highlights}\label{sec:highlights}

This section singles out a few of our results and techniques that seem
especially useful or motivating for follow-up work.
First, recall that we prove that $O(\eps^{-3/2} \log \eps^{-1})$
samples from an unknown MHR distribution 
--- or more generally, an unknown $\alpha$-strongly regular
distribution~\cite{CR14} --- are sufficient to achieve
expected revenue at least $1-\eps$ times that of the monoply price.
Remarkably, this is fewer than the $\approx \eps^{-2}$ samples
necessary to accurately estimate the expected revenue obtained by even
a single fixed price for such a distribution!\footnote{It is well
  known (e.g.~\citet[Lemma 5.1]{ABbook}) that, given a coin that either
  has bias $\tfrac{1}{2}-\eps$ or bias $\tfrac{1}{2}+\eps$,
  $\Omega(\eps^{-2} \log \tfrac{1}{\delta})$ coin flips are necessary
  to distinguish between the two cases with probability at least
  $1-\delta$.  The lower bound is information-theoretic and applies to
  arbitrarily sophisticated   learning methods.
This sample complexity lower bound translates
  straightforwardly to the 
  problem of estimating, by any means, the expected revenue of a fixed
  price for an 
  unknown MHR distribution up to a factor of $(1\pm\eps)$.}
%This fact extends the lesson originally gleaned from the
%$\tfrac{1}{2}$-approximation guarantee for a single simple:
%near-optimal revenue maximization requires far less information than
%accurate learning of a distribution.
In this sense, we prove that near-optimal revenue-maximization is
strictly easier than accurately learning 
even very simple statistics of the underlying distribution.
%Two ideas allow our upper bound proof to escape the ``learning
%bottleneck.''
The most important idea in our upper bound
is that, because of the structure of the
revenue-maximization problem, the estimation errors of different
competing prices are usefully correlated.
For example, if the estimated expected revenue of the true monopoly
price is significantly less than its actual expected revenue (because
of a higher-than-expected number of low samples), then
this probably also holds for prices that are relatively
close to the monopoly price.  Moreover, these are precisely the incorrect
prices that an algorithm is most likely to choose by mistake.
The second ingredient is the fact that MHR distributions have strongly
concave ``revenue curves,'' and this limits how many distinct prices
can achieve expected revenue close to that of the monopoly price.
%An intriguing open direction is to identify additional classes of
%problems with an unknown input distribution for which optimization is
%strictly easier than learning.

Second, recall that we prove
%We complement these upper bounds with
essentially matching lower bounds for all of our sample complexity upper
bounds.  For example, there is no $(1-\eps)$-approximate pricing
strategy (deterministic or randomized) for MHR distributions when $m =
o(\eps^{-3/2})$ or for regular distributions when $m = o(\eps^{-3})$.
For both of these lower bounds, we reduce the existence of a
$(1-\eps)$-optimal pricing strategy to that of a classifier that
distinguishes
between two similar distributions.  We borrow methodology from the
differential privacy literature to construct two distributions with
small KL divergence and disjoint sets of near-optimal prices, and use
Pinsker's inequality to derive the final sample complexity lower bounds.
This lower bound approach is novel in the context of auction theory
and we expect it to find further applications.
% in auction theory.

Third, we offer a simple and novel approach for reasoning about
irregular distributions.
We noted above the problematic irregular distributions that place a
very low probability on a very high value.  Regularity can also fail
for more ``reasonable'' distributions, such as mixtures of common
distributions.
In \pref{sec:gen}, we consider a benchmark $R^*_\delta$ defined
as the maximum expected revenue achievable for the underlying
distribution using a price that sells with probability at
least~$\delta$, and prove essentially tight sample complexity bounds
for approximating this benchmark.
As a special case, if for every distribution in $\D$ the monopoly
price sells with probability at least $\delta$ --- as is the case for
sufficiently small $\delta$ and typical ``reasonable'' distributions, even
irregular ones --- then
approximating $R^*_{\delta}$ is equivalent to approximating the
optimal revenue.
Even if not all distributions of $\D$ satisfy this property, this
benchmark enables parameterized sample complexity bounds that do not
require blanket distributional restrictions such as regularity.
%As a majority of the literature on
%approximately Bayesian-optimal auctions confines attention to regular
%distributions --- see~\cite{Hartline} for both the rule and some
%exceptions --- w
We believe that this parameterized approach will find
more applications.\footnote{See e.g.~\citet[Appendix D]{HR14},
  \citet[Chapter 4]{hartline},
  and~\citet{SS13} for alternative 
  approaches to parameterizing irregularity.}

\subsection{Further Related Work}

We already mentioned the related work of \citet{CR14}; the present
work follows the same formalism.
Specializing the results in~\citet{CR14} to the sample complexity
questions that we study here yields much weaker results than the ones we
prove --- only a lower bound of $\eps^{-1/2}$ and an upper bound of
$\eps^{-c}$ for a large constant $c$.
Two of the upper bounds in \pref{tab:onesample} follow from
previous work.  The upper bound of $O(\eps^{-3} \log \eps^{-1})$ for
regular distributions was proved in~\citet{DRY10}.  (They also proved 
a bound of $O(\eps^{-2} \log \eps^{-1})$ for MHR distributions,
which is subsumed by our nearly tight bound of $O(\eps^{-3/2} \log
\eps^{-1})$.)
The upper bound of $O(H\eps^{-2} \log H\eps^{-1})$ for bounded
valuations can be deduced from~\citet{BBHM05}.\footnote{The
  paper by~\citet{BBHM05} studies a seemingly different problem ---
  the design of digital good auctions with $n$ buyers in a
  prior-free setting (with bounded valuations).  But if one
  instantiates their model with 
  bidders with i.i.d.\ valuations from a distribution $\dist$, then
  their performance analysis of their RSO mechanism essentially
  gives a performance guarantee for the empirical monopoly price for $\dist$ with 
  $n/2$ samples, relative to the expected revenue of the monopoly
  price with a single bidder.}
We emphasize that,
in addition to our new upper bound results in the large-sample
regime, there is no previous work on sample complexity lower bounds
for our pricing problem nor on the best-possible approximation given a
single sample.

%We also mentioned the work of \citet{DRY10}, which through the
%design of prior-independent mechanisms indirectly studied a few of the
%same problems studied here.  The relevance of~\cite{DRY10} to the
%problems we study is detailed in Tables~\ref{tab:many}
%and~\ref{tab:onesample}.  \citet{DRY10} also prove a sample complexity
%upper bound of 

%The second set of problems that we study, the best-possible
%approximation given a single sample, is not considered at
%all in~\cite{CR14}.

%Our improved results for MHR distributions show that
%the upper bounds in \citet{DRY10} are not optimal.  No lower bounds are
%considered in \citet{DRY10} for any of the problems that we study.

%Recent work by \citet{CR14} study the number of samples that are
%necessary and sufficient for a $(1-\eps)$-approximation in a
%multi-agent setting, which is therefore more complex than the one
%studied here.
%Specifically \citet{CR14} consider a single-item auction with $k$
%non-identical bidders.
%The main lower bound in \citet{CR14} is that the number of samples
%(per bidder) must grow polynomially with the number of bidders $k$.
%In terms of $\eps$, the lower bound in \citet{CR14} only implies a
%necessary dependence of $\eps^{-1/2}$, which is much weaker than the
%lower bounds of $\eps^{-3/2}$ and $\eps^{-3}$ that we prove here.

There are many less related previous works that also use the idea of
independent samples in the context of auction design.
For example, some previous works study the asymptotic (in the
number of samples) convergence of an auction's revenue to the optimal
revenue,
% under different limits, 
without providing any uniform sample complexity bounds.
See \citet{N03}, \citet{segal},
\citet{BV03}, and \citet{G+06} for several examples.
Some recent and very different uses of samples in auction design include
\citet{FHHK14}, who use samples to extend the \citet{CM85} theorem to partially known valuation
distributions, and \citet{CHN14}, who design auctions
that both have near-optimal revenue and enable accurate inference about
the valuation distribution from samples.
%None of these papers quantify uniform sample complexity as we do here,
%and none consider how to optimally use a limited number of samples.

%{\bf (Zhiyi: We use $F$ to denote the distribution in the intro and use $D$ in the rest of the paper.. I'll fix it in my next pass..)}

\section{Preliminaries}
\label{sec:prelim}

%\subsection{Bayesian Auction Design Basics}
%Let the be a selling with a single item for sale, and let there be a single buyer...

Suppose the buyer's value is drawn from a publicly known distribution $D$ whose support is a continuous interval. 
Let $F$ be the c.d.f.\ of $D$.
If $F$ is differentiable, let $f$ be the p.d.f.\ of $D$. 
Let $q(v) = 1 - F(v)$ be the quantile of value $v$, i.e., the sale probability of reserve price $v$.
Let $v(q)$ be the value with quantile $q$. 

%In this work, we study two groups of distributions: distributions that satisfy standard small-tail assumptions such as regularity, monotone hazard rate, etc., which we will explain in more details later, and general distributions.

%\subsubsection*{Small-Tail Distributions}

The first set of distributions we study are those satisfying standard
small-tail assumptions such as regularity, monotone hazard rate, and
$\alpha$-strong regularity~\cite{CR14}.
We explain these assumptions in more detail next.
For these distributions, we assume $F$ is differentiable and $f$ exists.

Let $R(q) = qv(q)$ be the revenue as a function over the quantile space. We have
\[R'(q) = v(q) + q \frac{d v}{d q} = v - \frac{q(v)}{f(v)} \enspace.\]

The {\em virtual valuation function} is defined to be $\phi(v) = v - \tfrac{1 - F(v)}{f(v)} = R'(q)$.
A distribution $D$ is {\em regular} if for all value $v$ in its support,
\begin{equation}
\label{eq:RegDef}
\frac{d \phi}{d v} \ge 0 \enspace.
\end{equation}
Note that $v(q)$ is decreasing in $q$. 
A distribution is regular iff $R'(q) = \phi(v)$ is decreasing in $q$ and, thus, $R(q)$ is concave.
So $R(q)$ is maximized when $R'(q) = \phi(v(q)) = 0$.
Let $q^*$ and $v^* = v(q^*)$ be the revenue-optimal quantile and reserve price respectively.

A distribution $D$ has {\em monotone hazard rate} (MHR) if for all values $v$ in its support,
\begin{equation}
\label{eq:MHRDef}
\frac{d \phi}{d v} \ge 1 \enspace.
\end{equation}

%The exponential distributions are commonly considered to be the
%worst-case MHR distributions. 
%Consider the exponential distribution $D^\mhr$ with c.d.f.\ $F^\mhr(v)
%= 1 - e^{-v}$. Let $R^\mhr$ be the revenue curve of $D^\mhr$. We have 
%\[ R^\mhr(q) = - q \ln q \]
%Let $v^\mhr(q)$ be the value corresponding to quantile $q$
%w.r.t.\ this distribution. 
%$R^\mhr(q)$ is maximized at $\frac{1}{e}$.
%It is known that $R^\mhr$ minimizes the sale probability at the
%monopoly price among all MHR distributions. 

\begin{lemma}[\citet{HMS08}]
\label{lem:peak}
For every MHR distribution, $q^* \ge \frac{1}{e}$.
\end{lemma}

\citet{CR14} defined {\em $\alpha$-strong regular} distributions
to interpolate between \pref{eq:RegDef} and \pref{eq:MHRDef}:
%A distribution is $\alpha$-strongly regular if
\begin{equation}
\label{eq:StrRegDef}
\frac{d \phi}{d v} \ge \alpha \enspace.
\end{equation}
Many properties of MHR distributions carry over to $\alpha$-strongly
regular distributions with different constants.  For example:
%For instance, the
%monopoly price of any $\alpha$-strongly regular distribution has at
%least constant sale probability, where the constant depends on
%$\alpha$: 
\begin{lemma}[\citet{CR14}]
\label{lem:StrRegPeak}
For any $\alpha$-strongly regular distribution, $q^* \ge \alpha^{1/(1-\alpha)}$.
\end{lemma}

\begin{comment}
The characterization of MHR distributions via differential equation
also extends to $\alpha$-regular distributions. A distribution is
$\alpha$-strongly regular iff
\begin{equation}
\label{eq:StrRegCha}
R''(q) \le - \frac{\alpha}{q^2} \big( R(q) - q R'(q) \big) \enspace.
\end{equation}
\end{comment}

%\subsubsection*{General Distributions}

%We also seek to obtain positive results that hold for arbitrary
%distributions.  However, there is no uniform sample complexity upper
%bound when the benchmark is the optimal revenue:  if the optimal
%revenue comes from setting a high price with tiny sale probability,
%e.g., consider the family of distributions that take on a value $M^2$
%with probability $\tfrac{1}{M}$ and 0 with probability
%$1-\tfrac{1}{M}$, then we need at least $M$ samples to learn the
%existence of such a high-value in the support. 
%
%To remedy this issue, 
To reason about general (irregular) distributions, we require an
alternative benchmark (recall the Introduction).  We propose
%we propose an alternative benchmark
\[ R^*_{\delta} = \max_{q \ge \delta} q v(q) \enspace, \]
the optimal revenue if we only consider reserve prices with sale
probability at least $\delta$.
Here, we expect the sample complexity to depend on both $\epsilon$ and
$\delta$.

\begin{comment}
{\bf ******}

I'm not completely sure how to motivate ignoring prices with tiny sale probability. Here are some potential reasonings:

\begin{itemize}
\item The above example shows we simply don't have enough data to compete with the optimal price if it has tiny sale probability
\item The seller may want to maximize revenue while maintaining certain market share. So prices with tiny sale probability are undesirable. 
\item This benchmark captures a key property of the small-tale assumptions: it's easy to see that $R^*_{\epsilon} \ge (1 - \epsilon) R^*$ for regular distributions due concavity of their revenue curves; \pref{lem:peak} implies $R^*_{1/e} = R^*$ for MHR distributions. Indeed, our positive results for general distributions and these two facts about regular and MHR distributions allow us to rederive the $\tilde{O}(1/\epsilon^3)$ and $\tilde{O}(1/\epsilon^2)$ upper bounds by \citet{DRY10} for regular and MHR distributions respectively.
\end{itemize}

{\bf ******}

Another natural way to make general distributions tractable is to consider the bounded support case, where the buyer's value is between $1$ and $H$.
In this case, the sample complexity depends on both $\epsilon$ and $H$. 
%It turns out the techniques we need for the bounded support case is identical to the previous case.

\end{comment}

\section{Asymptotic Upper Bounds}
\label{sec:mhrmanyub}

We now present our positive results in the asymptotic regime.
%Our results are based on the empirical reserve pricing algorithm by
%\citet{DRY10}.

\begin{definition}
Given $m$ samples $v_1 \ge v_2 \ge \dots \ge v_m$, the {\em empirical
  reserve} is
\[ \argmax_{i \ge 1} i \cdot v_i. \]
\end{definition}

\noindent
If we only consider $i \ge c m$ for some parameter $c$, it is called
the {\em $c$-guarded empirical reserve}.
%\citet{DRY10} showed that the $\epsilon$-guarded empirical reserve
%with $\Theta(\epsilon^{-3} \log \epsilon^{-1})$ samples is $(1 -
%\epsilon)$-approximate for regular distributions.
%We will present a matching lower bound (up to a log factor) in
%\pref{sec:manysamplelb}.
%They also showed that for MHR distributions, the empirical reserve with $\Thet%a(\epsilon^{-2} \log \epsilon^{-1})$ samples suffices.

\subsection{MHR Upper Bound}

%In this section, we will improve the upper bound for MHR distributions
%to using only $\Theta(\epsilon^{-3/2} \log \epsilon^{-1})$ samples.
We next prove the following.

\begin{theorem}
\label{thm:mhrmanyub}
The empirical reserve with $m = \Theta(\epsilon^{-3/2} \log \epsilon^{-1})$ samples is $(1 - \epsilon)$-approximate for all MHR distributions.
\end{theorem}
We also give a matching lower bound (up to the log factor) in
\pref{sec:manysamplelb}.

For simplicity of presentation, we prove \pref{thm:mhrmanyub} for
the $\tfrac{1}{e}$-guarded empirical reserve.
(Recall that $q^* \ge \tfrac{1}{e}$ for MHR distributions.)
The unguarded version is similar but requires some extra care on the
small quantiles.
% (also see relevant discussions by \citet{DRY10}).

To show \pref{thm:mhrmanyub}, we use two properties of MHR distributions.
First, the optimal quantile of an MHR distribution is at least $e^{-1}$ (\pref{lem:peak}).
Second, %the following lemma proves that
the revenue decreases quadratically in how much the reserve
price deviates from the optimal one in quantile space, which we
formulate as the following lemma.% (see Appendix for the proof).

\begin{lemma}
\label{lem:mhrquadratic}
For any $0 \le q' \le 1$, we have $R(q^*) - R(q') \ge \frac{1}{4} (q^* - q')^2 R(q^*)$.
\end{lemma}

\begin{proof}%{Lemma}{lem:mhrquadratic}
There are three cases depending on the relation between $q'$ and $q^*$: $q' > q^*$, $q' = q^*$, and $q' < q^*$.
The second case, i.e., $q' = q^*$, is trivial.
Next, we prove the other two cases separately.

\medskip

First, consider the case when $q' > q^*$.
By the optimality of $q^*$, for any $q$ s.t.\ $q^* \le q \le q'$, we have $q v(q) \le q^* v(q^*)$ and, thus,
\[
v(q) \le \frac{q^*}{q} v(q^*)
\]
Further, by the MHR assumption that $\tfrac{d \phi(v)}{dv} \ge 1$, for any $q^* \le q \le q'$, we have
\[
\phi(v(q)) \le \phi(v(q^*)) + v(q) - v(q^*) = v(q) - v(q^*)
\]
Combining with the above inequality that lower bounds $v(q)$, we get that
\[
\phi(v(q)) \le \frac{q^* - q}{q} v(q^*)
\]
Therefore, we get that
\[
R(q^*) - R(q') = \int^{q'}_{q^*} - R'(q) dq = \int^{q'}_{q^*} - \phi(v(q)) dq \ge \int^{q'}_{q^*} \frac{q - q^*}{q} v(q^*) dq
\]
Note that $\frac{q - q^*}{q} \ge 0$ for any $q' \le q \le q^*$.
Moreover, for any $q \ge \tfrac{q' + q^*}{2}$, we have $\frac{q - q^*}{q} \ge \frac{q' - q^*}{q' + q^*}$.
Hence, we further drive the following inequality
\[
R(q^*) - R(q') \ge \int^{q'}_{\frac{q' + q^*}{2}} \frac{q' - q^*}{q' + q^*} v(q^*) dq = \frac{(q' - q^*)^2}{2(q' + q^*)} v(q^*) = \frac{(q' - q^*)^2}{2 q^*(q' + q^*)} R(q^*) \enspace.
\]
Then lemma then follows from that $0 \le q', q^* \le 1$.

\medskip

%The case when $q' > q^*$ is almost identical.
Next, we consider the case when $q' < q^*$.
The high-level proof idea of this case is similar to the previous case, but requires some subtle changes in the inequalities.
For completeness, we include the proof below

By concavity of the revenue curve, for any $q' \le q \le q^*$, we have
\[
q v(q) \ge \frac{q - q'}{q^* - q'} q^* v(q^*) + \frac{q^* - q}{q^* - q'} q' v(q') \enspace.
\]
Dividing both sides by $q$, we have
\[ v(q) \ge \frac{q^* v(q^*) - q' v(q')}{q^* - q'} + \frac{q^* q'}{q
  (q^* - q')} \big( v(q') - v(q^*) \big). \]
%\[ v(q) \ge \frac{q - q'}{q^* - q'} \frac{q^* - q}{q} v(q^*) + \frac{q^* - q}{q^* - q'} \frac{q' - q}{q} v(q') \]
%\[ v(q) \ge \frac{q - q'}{q^* - q'} \frac{q^* - q}{q} \big( v(q^*) - v(q') \big) \]
%\[ v(q) \ge \frac{q - q'}{q^* - q'} v(q^*) + \frac{q^* - q}{q^* - q'} v(q') \]
Further, by the MHR assumption,
\[
\phi(v(q)) \ge \phi(v(q^*)) + v(q) - v(q^*) ~=~ v(q) - v(q^*) \enspace.
\]
Note that the direction of the inequality is the opposite of its counterpart in the previous case.
This is because we have $v(q) > v(q^*)$ in this case 
%YM EC15
(as oppose to $v(q) < v(q^*)$
%(as oppose to $v(q) > v(q^*)$ 
as in the previous case.)
Combining with the above inequality that lower bounds $v(q)$, we get that
\begin{eqnarray*}
\phi(v(q)) & \ge & \frac{q^* v(q^*) - q' v(q')}{q^* - q'} + \frac{q^* q'}{q (q^* - q')} \big( v(q') - v(q^*) \big) - v(q^*) \\
& = & \frac{q' (q^* - q)}{q (q^* - q')} \big( v(q') - v(q^*) \big) ~
\ge ~ \frac{q' (q^* - q)}{q^* (q^* - q')} \big( v(q') - v(q^*) \big) \enspace,
\end{eqnarray*}
where the last inequality is due to $q \le q^*$.
Hence, we have
\begin{eqnarray}
R(q^*) - R(q') & = & \int^{q^*}_{q'} R'(q) dq ~ = ~ \int^{q^*}_{q'} \phi(v(q)) dq ~ \ge ~ \int^{q^*}_{q'} \frac{q' (q^* - q)}{q^* (q^* - q')} \big( v(q') - v(q^*) \big) dq \notag \\
& = & \frac{q'}{2 q^*} (q^* - q') \big( v(q') - v(q^*)
\big). \label{eq:manyprepeak1}
\end{eqnarray}
On the other hand, we have
\begin{equation}
\label{eq:manyprepeak2}
R(q^*) - R(q') = q^* v(q^*) - q' v(q').
\end{equation}
Taking the 
%YM EC15
linear
%lienar 
combination $\frac{2 q^*}{3 q^* - q'} \cdot
\pref{eq:manyprepeak1} + \frac{q^* - q'}{3 q^* - q'} \cdot
\pref{eq:manyprepeak2}$, we have
%
%\[ \frac{q^* + q'}{2 q'} \big( R(q') - R(q^*) \big) \le - \frac{1}{2} \frac{(q^* - q')^2}{q^* q'} R(q^*) \]
\[
R(q^*) - R(q') \ge \frac{(q^* - q')^2}{3 q^* - q'} v(q^*) =
\frac{1}{q^* (3 q^* - q')} (q^* - q')^2 R(q^*) \ge \frac{1}{3} (q^* -
q')^2 R(q^*) \enspace,
\]
where the last inequality holds because $0 \le q^*, q' \le 1$.
\end{proof}

Next we show how to use the lemma (and additional ideas) to prove \pref{thm:mhrmanyub}.

\begin{proof}[Proof of \pref{thm:mhrmanyub}]
We first show that for any two samples $v_1$ and $v_2$ with
quantiles $q_1$, $q_2$ such that either $q_1 < q_2 < q^*$ or $q^* <
q_1 < q_2$, if the revenue of one of them is at least
$(1-\tfrac{\epsilon}{2})$ times smaller than that of the other, e.g.,
$q_1 < q_2 < q^*$ and $v_1 q_1 < (1 - \epsilon) v_2 q_2$, then with
probability at least $1 - o(\tfrac{1}{m^2})$ the algorithm would
choose $v_2$ over $v_1$.
Further, with high probability, there is at least one sample that is $\tfrac{\epsilon}{2}$-close to $q^*$ in quantile space both among samples with quantile at least $q^*$ and among those with quantile at most $q^*$.
By concavity of the revenue curve, such samples are $(1-\tfrac{\epsilon}{2})$-optimal.
So the theorem follows from union bound.

Let us focus on the case when $q_1 < q_2 < q^*$ and $v_1 q_1 < (1 - \epsilon) v_2 q_2$ as the other case is almost identical.
Suppose $R(q_1) = (1 - \Delta) R(q_2)$ and $q_1 = q_2 - \delta$.
By concavity of the revenue curve and \pref{lem:mhrquadratic}, we have
\[ R(q_2) - R(q_1) \ge R(q^*) - R(q^* - q_2 + q_1) \ge \frac{1}{4} (q_2 - q_1)^2 R(q^*) \ge \frac{1}{4} (q_2 - q_1)^2 R(q_2). \]
So we have $\Delta = \Omega(\delta^2)$.

Let $\tilde{q}_i m$ be the number of samples with value at least $v_i$, $i = 1, 2$.
The goal is to show that $\tilde{q}_1 v_1 < \tilde{q}_2 v_2$ with probability at least $1 - o(\tfrac{1}{m^2})$.

Note that the straightforward argument does not work because that we would need $\tilde{\Omega}(\epsilon^{-2})$ samples to estimate $q_i$ up to a $(1 - \epsilon)$ factor.
Before diving into the technical proof, let us explain informally how to get away with fewer samples.
A bad scenario for the straightforward argument is when, say,
$\tilde{q}_1 > (1 + \Delta) q_1$ and $\tilde{q}_2 < (1 - \Delta)
q_2$.
We observe that such a bad scenario is very unlikely due to correlation between $\tilde{q}_1$ and $\tilde{q}_2$: the samples used to estimate $q_1$ and $q_2$ are the same; those that cause the algorithm to overestimate $q_1$ also contribute to the estimation of $q_2$; for the bad scenario to happen, it must be that the number of samples $q_1$ and $q_2$ is much smaller than its expectation (as we will formulate as \pref{eq:mhrmany1}), whose probability is tiny.

Now we proceed with the formal proof.
Since we consider the $\frac{1}{e}$-guarded empirical reserve, $q_1, q_2 \ge
\frac{1}{e}$.
%Further, recall that $q^* \ge \frac{1}{e}$ (\pref{lem:peak}).
By the Chernoff bound, with $\Theta(\epsilon^{-3/2} \log
\epsilon^{-1})$ samples, we have $\tilde{q}_i \ge (1 - \epsilon^{3/4})
q_i = \Omega(1)$ and $\tilde{q}_i \le (1 + \epsilon^{3/4}) q_i$ with
high probability.

If $\tilde{q}_1 v_1 \ge \tilde{q}_2 v_2$, then %we have
\[ \frac{\tilde{q}_1}{\tilde{q}_2} \ge \frac{v_2}{v_1} = \frac{R(q_2)}{R(q_1)} \frac{q_1}{q_2} = (1 - \Delta)^{-1} \frac{q_1}{q_2} = \frac{q_1}{q_2} + \Omega(\Delta) \enspace.\]
So
\begin{eqnarray*}
\tilde{q}_2 - \tilde{q}_1 & = & \left( 1 - \frac{\tilde{q}_1}{\tilde{q}_2} \right) \tilde{q}_2 ~ \le ~ \left( 1 - \frac{\tilde{q}_1}{\tilde{q}_2} \right) (1 + \epsilon^{3/4}) q_2 \\
& \le & \left( 1 - \frac{q_1}{q_2} - \Omega(\Delta) \right) ( 1 +
\epsilon^{3/4} ) q_2 ~ = ~ q_2 - q_1 + \delta \epsilon^{3/4} -
\Omega(\Delta).
\end{eqnarray*}
Since $\delta = O(\sqrt{\Delta})$ and $\Delta \ge \frac{\epsilon}{2}$,
we have $\delta \epsilon^{3/4} = o(\Delta)$. So
\begin{equation}
\label{eq:mhrmany1}
\tilde{q}_2 - \tilde{q}_1 \le q_2 - q_1 - \Omega(\Delta) \enspace.
\end{equation}
That is, the number of samples that fall between $q_1$ and $q_2$ is smaller than its expectation by at least $\Omega(\Delta m)$.
By the Chernoff bound, the probability of this event is at most $\exp
\big( - \frac{\Delta^2 m}{\delta} \big)$.
Recall that $\Delta = \Omega(\delta^2)$, $\Delta \ge
\frac{\epsilon}{2}$, and $m = \Theta(\epsilon^{-3/2} \log
\epsilon^{-1})$. So this probability is at most $\exp(-\Omega(\log
\epsilon^{-1}) ) = o(\tfrac{1}{m^2})$ with an appropriate choice of
parameters.
\end{proof}

\subsection{$\alpha$-Strongly Regular Upper Bound}

%The techniques we used to show improved sample complexity upper bound
%for MHR distributions
Our proof of \pref{thm:mhrmanyub} can be extended to $\alpha$-strongly regular
distributions with $\alpha > 0$.
We present the formal statement and sketch the necessary changes below.

\begin{theorem}
\label{thm:strregmanyub}
The empirical reserve with $m = \Theta(\epsilon^{-3/2} \log \epsilon^{-1})$ samples is $(1 - \epsilon)$-approximate for all $\alpha$-strongly regular distributions, for a constant $\alpha>0$.
\end{theorem}

The proof of \pref{thm:mhrmanyub} relies on two properties: the
monopoly price having at least constant sale probability
(\pref{lem:peak}), and strict concavity of the revenue curve at the
monopoly quantile (\pref{lem:mhrquadratic}).
The proof of \pref{thm:strregmanyub} is identical, modulo using weaker
versions of the lemmas.
Specifically, we will replace \pref{lem:peak} by
\pref{lem:StrRegPeak}, and \pref{lem:mhrquadratic} by the following
lemma, whose
proof is almost identical to that of \pref{lem:mhrquadratic}.

\begin{lemma}
\label{lem:strregquadratic}
For any $q' \ne q^*$, $R(q^*) - R(q') \ge \frac{\alpha}{3} (q^* - q')^2 R(q^*)$.
\end{lemma}

In the Appendix we use \pref{lem:strregquadratic} to prove
\pref{thm:strregmanyub}.

\subsection{General Upper 
%YM EC15
Bounds
%Bound
}\label{sec:gen}

Next, we present a sample complexity upper bound for general
distributions using $R^*_\delta$ as a benchmark.
Recall that $R^*_\delta$ is the optimal revenue by prices with sale probability at least $\delta$.

\begin{theorem}
\label{thm:generalmanyub}
The $\tfrac{\delta}{2}$-guarded empirical reserve with $m = \Theta(\delta^{-1} \epsilon^{-2} \log (\delta^{-1} \epsilon^{-1}) )$ gives revenue at least $(1 - \epsilon) R^*_{\delta}$ for all distributions.
\end{theorem}

\begin{proof}[Proof Sketch]
The proof is standard so we present only a sketch here.
Let $q^*_{\delta} = \argmax_{q \ge \delta} q v(q)$ be the optimal reserve price with sale probability at least $\delta$.
With high probability, there exists at least one sampled price with quantile between $(1-\tfrac{\epsilon}{3}) q^*$ and $q^*$: this price has revenue at least $(1 - \tfrac{\epsilon}{3}) R^*_{\delta}$. Further, since $q^* \ge \delta$, this price has rank at least $\tfrac{\delta}{2}$ (among sampled prices) with high probability and thus is considered by the empirical reserve algorithm.
Finally, with high probability, any sampled price with rank at least $\tfrac{\delta}{2}$ has sale probability at least $\tfrac{\delta}{4}$; for prices with sale probability at least $\tfrac{\delta}{4}$, the algorithm estimates their sale probability up to a $1 - \tfrac{\epsilon}{3}$ factor with high probability with $m = \Theta(\delta^{-1} \epsilon^{-2} \log (\delta^{-1} \epsilon^{-1}) )$ samples.
The theorem then follows.
\end{proof}

%We remark that the $O( \epsilon^{-3} \log \epsilon^{-1})$ upper bound
%for regular distributions by \citet{DRY10} follows by
%\pref{thm:generalmanyub} and that $R^*_{\epsilon} \ge (1 - \epsilon)$
%for regular distributions (due to concavity of their revenue curves).
%Their $O( \epsilon^{-2} \log \epsilon^{-1})$ upper bound for MHR
%distributions follows by \pref{thm:generalmanyub} and that $R^*_{1/e}
%= R^*$ (\pref{lem:peak}).
%%
We remark that one can also derive a bound for a single
sample (i.e., $m=1$) which guarantees expected revenue at least
$(\delta/2) R^*_{\delta}$.

%YM EC15
%\zhiyi{Is this a good place for this remark?}
%, which is interesting in the case that the optimal price sells with constant probability.

Note that for distributions with support $[1, H]$, the optimal sale
probability is at least $1/H$. So we have the following theorem as a
direct corollary of \pref{thm:generalmanyub}.  This bound can also be
deduced from 
%YM EC 15
%previous work
\cite{BBHM05}; we include it for
completeness.

\begin{theorem}
\label{thm:generalmanyub2}
The empirical reserve with $m = \Theta(H \epsilon^{-2} \log (H \epsilon^{-1}) )$ samples is $(1 - \epsilon)$-approximate for all distributions with support $[1, H]$.
\end{theorem}

\section{Asymptotic Lower Bounds}
\label{sec:manysamplelb}

This section gives asymptotically tight (up to a log
factor) sample complexity lower bounds.
These lower bounds are information-theoretic and apply to all possible
pricing strategies, including randomized strategies.
We first present a general framework for proving sample complexity
lower bounds, and then instatiate it for each of the classes of
distributions listed in \pref{tab:many}.
%All of the lower bounds in this section hold for both
%deterministic and randomized pricing rules.
% for the pricing problem in \pref{sec:manylbframework}.
%Then, we will use this framework to derive the lower bounds for
%regular distributions in \pref{sec:regularmanylb} and for MHR
%distributions in \pref{sec:mhrmanylb}.

\subsection{Lower Bound Framework: Reducing Pricing to Classification}
\label{sec:manylbframework}

The high-level plan is to reduce the pricing problem to a
classification problem.
We will construct two distributions $D_1$ and $D_2$ and show that
given any pricing algorithm that is $(1 - \epsilon)$-approximate for
both $D_1$ and $D_2$, we can construct a classification algorithm that
can distinguish $D_1$ and $D_2$ with constant probability, say,
$\frac{1}{3}$, using the same number of samples as the pricing
algorithm.
Further, we will construct $D_1$ and $D_2$ to be similar enough and
use tools from information theory to show a lower bound on the number
of samples needed to distinguish the two distributions.
%So the same lower bound holds for pricing algorithms.

\paragraph{Information Theory Preliminaries}

Consider two distributions $P_1$ and $P_2$ over a sample
space~$\Omega$. Let $p_1$ and $p_2$ be the density functions. The {\em 
  statistical distance} between $P_1$ and $P_2$ is:
\[ \delta(P_1, P_2) = \frac{1}{2} \int_{\Omega} \big| p_1(\omega) -
p_2(\omega) \big| d\omega. \]

In information theory, it is known (e.g., \citet{BJV04}) that no
classification algorithm $A : \Omega \rightarrow \{1, 2\}$ can
distinguish $P_1$ and $P_2$ correctly with probability strictly better
than $\frac{\delta(P_1, P_2) + 1}{2}$, i.e., there exists $i \in \{1,
2\}$, $\Pr_{\omega \sim P_i} \big[ A(\omega) = i \big] \le
\frac{\delta(P_1, P_2) + 1}{2}$.
This lower bound applies to arbitrary randomized classification
algorithms.
% {\bf (reference needed)}

Suppose we want to show a sample complexity lower bound of $m$.
Then we will let $P_i = D_i^m$ and upper bound $\delta(P_1, P_2)$.
However, the statistical distance is hard to bound directly when we
have multiple samples: $\delta(D_1^m, D_2^m)$ cannot be written as
function of $m$ and $\delta(D_1, D_2)$. In particular, the statistical
distance does not grow linearly with the number of samples.

In order to derive an upper bound on the statistical distance with
multiple samples, it is many times convenient to use the {\em
  Kullback-Leibler (KL) divergence}, which is defined as follows:
\[ \kl(P_1 \| P_2) = \E_{\omega \sim P_1} \left[ \ln \frac{p_1(\omega)}{p_2(\omega)} \right] \enspace. \]
In information theory, the KL divergence can be viewed as the
redundancy in the encoding in the case that the true distribution is
$P_1$ and we use the optimal encoding for distribution $P_2$.
One nice property of the KL divergence is that it is additive over
samples: if $P_1 = D_1^m$ and $P_2 = D_2^m$ are the
distributions over $m$ samples of $D_1$ and $D_2$, 
then the $KL$ divergence of $P_1$ and $P_2$ is $m$ times $\kl(D_1 \|
D_2)$.

We can relate the KL divergence to the statistical distance through
Pinsker's inequality~\cite{pinsker}, which states that:
\[ \textstyle \delta(P_1, P_2) \le \sqrt{\frac{1}{2} \kl( P_1 \| P_2 )} \]
By symmetry, we also have $\delta(P_1, P_2) \le \sqrt{\frac{1}{2} \kl(
  P_2 \| P_1 )}$, so
\[ \delta(P_1, P_2) \le \frac{1}{2} \sqrt{\kl( P_1 \| P_2 ) + \kl( P_2
  \| P_1 )}. \]
This implies that we can upper bound the statistical distance of $m$
samples from $D_1$ and $D_2$ by $\frac{1}{2} \sqrt{m \cdot ( \kl( D_1
  \| D_2 ) + \kl( D_2 \| D_1 ) )}$.
To get statistical distance at least, say, $\frac{1}{3}$, we need
$m=\frac{4}{9}\frac{1}{\kl( D_1 \| D_2 ) + \kl( D_2 \| D_1 )}$
samples.

\paragraph{Reducing Pricing to Classification}

Next, we present the reduction from pricing to classification.
Given a value distribution $D$ and $\alpha < 1$, its {\em
  $\alpha$-optimal price set} is defined to be the set of reserve
prices that induce at least $\alpha$ fraction of the optimal revenue.

\begin{lemma}
If value distributions $D_1$ and $D_2$ have disjoint $(1 - 3
\epsilon)$-approximate price sets, and there is a pricing algorithm
that is $(1 - \epsilon)$-approximate for both $D_1$ and $D_2$, then
there is an classification algorithm that distinguish $P_1$ and $P_2$
correctly with probability at least $\frac{2}{3}$, using the same
number of samples as the pricing algorithm.
\end{lemma}

\begin{comment}
\begin{proof}
Consider the following classification algorithm: for $i = 1, 2$, if
the pricing algorithm prices in the $(1 - 3 \epsilon)$-approximate
price set of $D_i$, then output $i$; otherwise, output $1$ or $2$
uniformly at random.
Since the pricing algorithm is $(1 - \epsilon)$-approximate for $D_i$,
the probability that it prices in the $(1 - 3 \epsilon)$-approximate
price set of $D_i$ when $D_i$ is the underlying distribution is at
least $\frac{2}{3}$.
The lemma then follows.
\end{proof}
\end{comment}

We omit the straightforward proof.
Note that to distinguish $P_1$ and $P_2$ correctly with probability at
least $\frac{2}{3}$, the statistical distance between $P_1$ and $P_2$
is least $\frac{1}{3}$.
So we have the following theorem.

\begin{theorem}
\label{thm:reduction}
If value distributions $D_1$ and $D_2$ have disjoint $(1 - 3
\epsilon)$-approximate price sets, and there is a pricing algorithm
that is $(1 - \epsilon)$-approximate for both $D_1$ and $D_2$, then
the algorithm uses at least $\frac{4}{9} \frac{1}{\kl( D_1 \| D_2 ) +
  \kl( D_2 \| D_1 )}$ samples.
\end{theorem}

\paragraph{A Tool for Constructing Distributions with Small KL Divergence}

Given \pref{thm:reduction}, our goal is to construct a pair of distributions with small relative entropy subject to having disjoint approximately optimal price sets.
Here we introduce a lemma from the differential privacy literature that is useful for constructing pairs of distributions with small KL divergence.

\begin{lemma}[Lemma III.2 of \citet{DRV10}]
\label{lem:dptrick}
If distributions $D_1$ and $D_2$ with densities $f_1$ and $f_2$
satisfy that $(1 + \epsilon)^{-1} \le \frac{f_1(\omega)}{f_2(\omega)}
\le (1 + \epsilon)$ for every $\omega \in \Omega$, then 
\[ \kl(D_1 \| D_2) + \kl(D_2 \| D_1) \le \epsilon^2 \enspace. \]
\end{lemma}

For completeness, we include the proof in the Appendix.

\begin{comment}
\begin{proof}
By the definition of KL divergence, we have
\begin{eqnarray*}
\kl(D_1 \| D_2) + \kl(D_2 \| D_1) & = & \int_\Omega \left[ p_1(\omega) \ln \frac{p_1(\omega)}{p_2(\omega)} + p_2(\omega) \ln \frac{p_2(\omega)}{p_1(\omega)} \right] d \omega \\
& = & \int_\Omega \left[ p_1(\omega) \left( \ln \frac{p_1(\omega)}{p_2(\omega)} + \ln \frac{p_2(\omega)}{p_1(\omega)} \right) + \big( p_2(\omega) - p_1(\omega) \big) \ln \frac{p_2(\omega)}{p_1(\omega)} \right] d \omega \\
& \le & \int_\Omega \bigg[ 0 + | p_2(\omega) - p_1(\omega) | \ln (1 + \epsilon) \bigg] d \omega
\end{eqnarray*}
The last inequality follows by $(1 + \epsilon)^{-1} \le \frac{f_1(\omega)}{f_2(\omega)} \le (1 + \epsilon)$.
Further, this condition also implies that
\[ | p_2(\omega) - p_1(\omega) | \le \big( (1 + \epsilon) - 1 \big) \min \{ p_1(\omega), p_2(\omega) \} = \epsilon \min \{ p_1(\omega), p_2(\omega) \} \enspace. \]
So we have
\[ \kl(D_1 \| D_2) + \kl(D_2 \| D_1) \le \epsilon \ln(1 + \epsilon) \int_\Omega \min \{ p_1(\omega), p_2(\omega) \} d \omega \le \epsilon \ln(1 + \epsilon) \le \epsilon^2 \enspace. \]
\end{proof}
\end{comment}

%There are many ways to tweak \pref{lem:dptrick} and its still correct with essentially the same argument.
%We present two such variants that are used in our analysis in the rest of this section and sketch the proofs.
The following two useful variants have similar proofs.

\begin{lemma}
\label{lem:dptrickregular}
If distributions $D_1$ and $D_2$ satisfy the condition in \pref{lem:dptrick}, and further there is a subset of outcomes $\Omega'$ such that $p_1(\omega) = p_2(\omega)$ for every $\omega \in \Omega'$, then
\[ \kl(D_1 \| D_2) + \kl(D_2 \| D_1) \le \epsilon^2 \big( 1 - p_1(\Omega') \big) \enspace. \]
\end{lemma}

\begin{comment}
\begin{proof}[Proof Sketch]
With this additional condition, we only need to integrate over $\Omega \setminus \Omega'$ in the proof of \pref{lem:dptrick}.
So the strengthened bound follows.
\end{proof}
\end{comment}

\begin{lemma}
\label{lem:dptrickmhr}
If distributions $D_1$ and $D_2$ satisfy that $(1 + \epsilon)^{-1} \le
\frac{f_1(\omega)}{f_2(\omega)} \le (1 + \epsilon)$ for every $\omega
\in \Omega$ and $(1 + \epsilon')^{-1} \le
\frac{f_1(\omega)}{f_2(\omega)} \le (1 + \epsilon')$ for any $\omega
\in \Omega' \subseteq \Omega$, then 
\[ \kl(D_1 \| D_2) + \kl(D_2 \| D_1) \le \epsilon^2 p_1(\Omega \setminus \Omega') + (\epsilon')^2 p_1(\Omega') \enspace. \]
\end{lemma}

\begin{comment}
\begin{proof}[Proof Sketch]
The proof is essentially the same as that of \pref{lem:dptrick}, except that we consider two integrations over $\omega \in \Omega \setminus \Omega'$ and $\omega \in \Omega'$ respectively.
%So the strengthened bound follows.
\end{proof}
\end{comment}

\subsection{Applications}

Inspired by the above lemmas, we will aim to construct $D_1$ and $D_2$
such that the densities of all values are close in the two
distributions.

%\subsection{General Lower Bound}
\paragraph{General Lower Bound}
As a warm-up case, we demonstrate how to use the above framework to
derive a tight (up to a log factor) sample complexity lower bound for
general distributions using $R^*_\delta$ as benchmark.
Recall that $R^*_\delta$ is the optimal revenue by prices with sale
probability at least $\delta$.

\begin{theorem}
\label{thm:generalmanylb}
Every pricing algorithm that guarantees at least $(1 - \epsilon)
R^*_\delta$ revenue for all distributions uses at least
$\Omega(\delta^{-1} \epsilon^{-2})$ samples.
\end{theorem}

\begin{proof}
Let $D_1$ and $D_2$ be two distributions with support $\{ H = \delta^{-1}, 2, 1 \}$:
$D_1$ takes value $H$ with probability $\tfrac{1 + 3 \epsilon}{H}$, $2$ with probability $\tfrac{1 - 3 \epsilon}{H}$, and $1$ with probability $1 - \tfrac{2}{H}$;
$D_2$ takes value $H$ with probability $\tfrac{1 - 3 \epsilon}{H}$, $2$ with probability $\tfrac{1 + 3 \epsilon}{H}$, and $1$ with probability $1 - \tfrac{2}{H}$.
Clearly, $D_1$ and $D_2$ have disjoint $(1-3\epsilon)$-approximate
price sets.
Further,
\[ \kl(D_1 \| D_2) = \kl(D_2 \| D_1) = \tfrac{1+3\epsilon}{H} \ln \tfrac{1+3\epsilon}{1-3\epsilon} + \tfrac{1-3\epsilon}{H} \ln \tfrac{1-3\epsilon}{1+3\epsilon} = \tfrac{6\epsilon}{H} \ln \tfrac{1+3\epsilon}{1-3\epsilon} = O(\tfrac{\epsilon^2}{H}) \enspace.\]
So the claim follows from \pref{thm:reduction}.
\end{proof}

The way we prove \pref{thm:generalmanylb} also implies a tight (up to
a log factor) sample complexity lower bound for distributions with
support in $[1, H]$. 

\begin{theorem}
\label{thm:generalmanylb2}
Any pricing algorithm that is $(1 - \epsilon)$-approximate for all distributions with support $[1, H]$ uses at least $\Omega(H \epsilon^{-2})$ samples.
\end{theorem}

%\subsection{Regular Lower Bound}
%\label{sec:regularmanylb}
\paragraph{Regular Lower Bound}
We now
%construct lower bound distributions for the regular case and
show that $(1 - \epsilon)$-approximate pricing for regular
distributions requires $\Omega(\epsilon^{-3})$ samples.
%, matching the
%upper bound by \citet{DRY10} up to a log factor.

\begin{theorem}
\label{thm:regmanylb}
Any pricing algorithm that is $(1 - \epsilon)$-approximate for all
regular distributions uses at least $\frac{(1 - 6\epsilon)^2}{486
  \epsilon^3} = \Omega(\frac{1}{\epsilon^3})$ samples.
\end{theorem}

This result implies, for example, we need at least $12$ samples to guarantee
$95$ percent of the optimal revenue, and at least $1935$ samples to
guarantee $99$ percent.

We next describe the two distributions that we use and explain the lower bound for regular distributions.%; the full proof of the MHR case is in the Appendix.
%\bigskip
%\subsection*{Proof of \pref{thm:regmanylb}}
%\begin{proof}%[ of \pref{thm:regmanylb}]
Let $D_1$ be the distribution with c.d.f.\ $F_1(v) = 1 - \frac{1}{v + 1}$ and
p.d.f.\ $f_1(v) = \frac{1}{(v + 1)^2}$. Let $\epsilon_0 = 3 \epsilon$
and let $D_2$ be the distribution with c.d.f.\
\[ F_2(v) = \left\{\begin{aligned}
& \textstyle 1 - \frac{1}{v + 1} & & \textstyle \textrm{if $0 \le v \le \frac{1 - 2\epsilon_0}{2\epsilon_0}$} \\
& \textstyle 1 - \frac{(1 - 2\epsilon_0)^2}{v - (1 - 2\epsilon_0)} & & \textstyle \textrm{if $v > \frac{1 - 2\epsilon_0}{2\epsilon_0}$}
\end{aligned}\right. \]
and p.d.f.\
\[ f_2(v) = \left\{\begin{aligned}
& \textstyle \frac{1}{(v + 1)^2} & & \textstyle \textrm{if $0 \le v \le \frac{1 - 2\epsilon_0}{2\epsilon_0}$} \\
& \textstyle \frac{(1 - 2\epsilon_0)^2}{(v - (1 - 2\epsilon_0))^2} & & \textstyle \textrm{if $v > \frac{1 - 2\epsilon_0}{2\epsilon_0}$.}
\end{aligned}\right. \]

%\subsection{Proof of \pref{thm:regmanylb}}

We have
\begin{equation}
\label{eq:regularclassification}
\frac{f_1}{f_2} = \left\{\begin{aligned}
& 1 & & \textrm{if $0 \le v \le \frac{1 - 2\epsilon_0}{2\epsilon_0}$} \\
& \frac{1}{(1 - 2\epsilon_0)^2} \frac{(v - (1 - 2\epsilon_0))^2}{(v + 1)^2} \in [(1 - 2\epsilon_0)^2, (1 - 2\epsilon_0)^{-2}] & & \textrm{if $v > \frac{1 - 2\epsilon_0}{2\epsilon_0}$.}
\end{aligned}\right.
\end{equation}
The revenue curves of $D_1$ and $D_2$ are summarized in \pref{fig:regmanylb}.
\begin{figure}
    \centering
    \includegraphics[width=.45\textwidth]{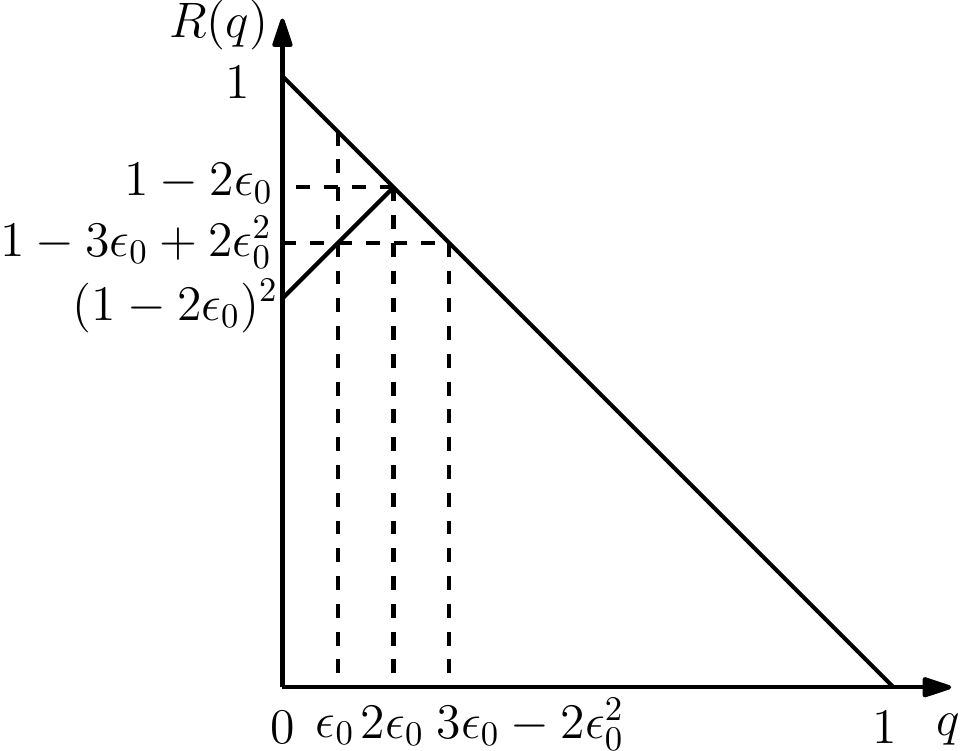}
    \caption{$D_1$ is the distribution with revenue curve $R_1$ that goes
      from $(0, 1)$ to $(1,  0)$. $D_2$ is identical to $D_1$ for
      quantiles from $2\epsilon_0$ to $1$; for quantiles from $0$ to
      $2\epsilon_0$, $D_2$'s revenue curve goes from $(0,
      (1-2\epsilon_0)^2)$ to $(2\epsilon_0, 1-2\epsilon_0)$.} 
    \label{fig:regmanylb}
\end{figure}

\begin{lemma}
\label{lem:klregular}
\[ \kl(D_1 \| D_2) + \kl(D_2 \| D_1) \le \frac{8 \epsilon_0^3}{(1 -
  2\epsilon_0)^2}. \] 
\end{lemma}

\begin{proof}
By \pref{eq:regularclassification}, we have $(1-2\epsilon_0) \le \frac{f_1}{f_2} \le (1-2\epsilon_0)^{-1}$.
Further, note that a $1 - 2\epsilon$ fraction (w.r.t.\ quantile) of
$D_1$ and $D_2$ are identical. 
The lemma follows from \pref{lem:dptrickregular}.
\end{proof}

Let $R_1$ and $R_2$ be the revenue curves of $D_1$ and $D_2$.
Let $R^*_1$ and $R^*_2$ be the corresponding optimal revenues.
The following lemmas follow directly from the definition of $D_1$ and $D_2$.

\begin{lemma}
$R_1(v) \ge (1 - \epsilon_0) R^*_1$ if and only if $v \ge
  \frac{1}{\epsilon_0} - 1$. 
\end{lemma}

\begin{lemma}
$R_2(v) \ge (1 - \epsilon_0) R^*_1$ if and only if
  $\frac{1}{\epsilon_0} - 3 + 2 \epsilon_0 \ge v \ge
  \frac{1}{2\epsilon_0 - \epsilon_0^2} - 1$. 
\end{lemma}

Recall that $\epsilon_0 = 3 \epsilon$. The $(1 -
3\epsilon)$-optimal price sets of $D_1$ and $D_2$ are disjoint. 
\pref{thm:regmanylb} follows from
\pref{thm:reduction} and \pref{lem:klregular}.

\paragraph{MHR Lower Bound}
We now turn to MHR distributions and show that $(1 -
\epsilon)$-approximate pricing for MHR distributions requires
$\Omega(\epsilon^{-3/2})$ samples.
%, matching the upper bound in
%\pref{thm:mhrmanyub} up to a log factor.

\begin{theorem}
\label{thm:mhrmanylb}
Any pricing algorithm that is $(1 - \epsilon)$-approximate for all MHR
distributions uses at least $\Omega(\epsilon^{-3/2})$ samples.
\end{theorem}

%\subsection*{Proof of \pref{thm:mhrmanylb}}

We again describe the two distributions used and defer the full proof
to the Appendix.
Let $D_1$ be the uniform distribution over $[1, 2]$.
Let $\epsilon_0 = c \epsilon$ where $c$ is a sufficiently large constant to be determined later.
Define $D_2$ by scaling up the density (of $D_1$) in $v \in [1 + \sqrt{\epsilon_0}, 2]$ by a factor of $1 + \frac{2 \epsilon_0}{1 - \sqrt{\epsilon_0}}$ and scaling down the density in $v \in [1, 1 + \sqrt{\epsilon_0}]$ by a factor $1 - 2\sqrt{\epsilon_0}$, i.e.,
\[ f_2(v) = \left\{\begin{aligned}
& 1 - 2 \sqrt{\epsilon_0} & & \textrm{if $1 \le v \le 1 + \sqrt{\epsilon_0}$} \\
& \textstyle 1 + \frac{2 \epsilon_0}{1 - \sqrt{\epsilon_0}} & & \textrm{if $1 + \sqrt{\epsilon_0} < v \le 2$.}
\end{aligned}\right. 
\]

We summarize the revenue curves of $D_1$ and $D_2$ in \pref{fig:mhrmanylb}.

\begin{figure}
\centering
\includegraphics[width=.48\textwidth]{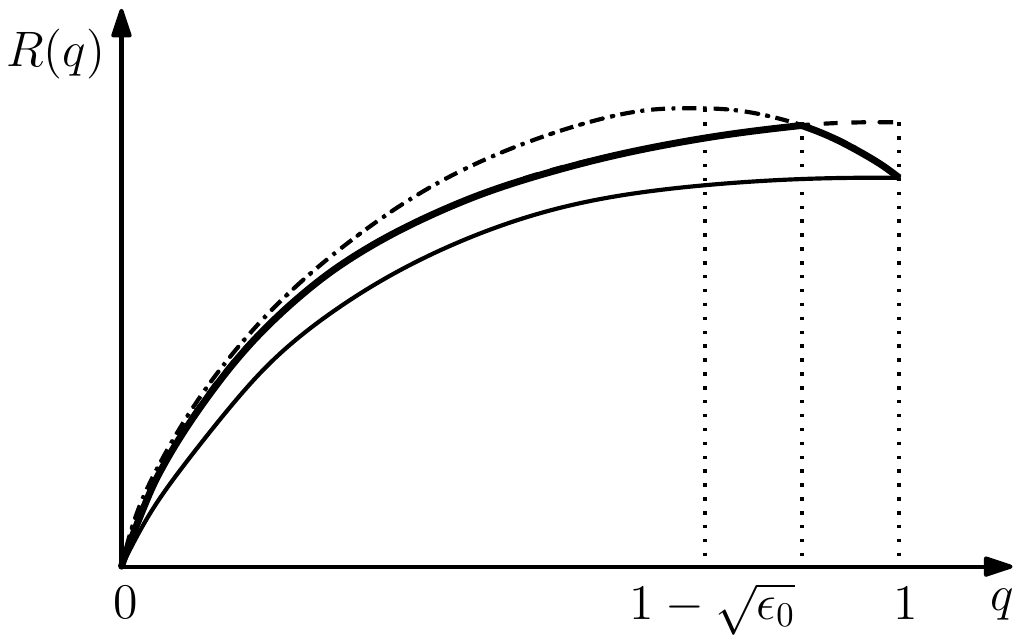}
\caption{$R_1$ (the lower solid curve) is a quadratic curve that peaks
  at $q = 1$, $R_1 = 1$, and passes through $q = 0$, $R_1 = 0$. To construct
  $R_2$, first draw the revenue curves of the uniform distributions
  over $\big[1 + \frac{2\epsilon_0}{1 - \sqrt{\epsilon_0} +
      2\epsilon_0}, 2\big]$ (the dashed curve) and $\big[1, 1 + \frac{1}{1
      - 2 \sqrt{\epsilon_0}}\big]$ (the dash-dotted curve). $R_2$ (the
  bold
  solid curve) is the lower envelope of the two curves.} 
\label{fig:mhrmanylb}
\end{figure}

\section{Single Sample Regime: Beating Identity Pricing for
MHR Distributions}
\label{sec:onesampleub}

This section considers deterministic 1-sample pricing strategies.
Recall from the Introduction that ``identity pricing,'' meaning
$p(\val_i) = \val_1$, has an approximation guarantee of $\tfrac{1}{2}$
for the class of regular distributions.  We show in the appendix that
there is no better 1-sample deterministic
pricing strategy for the class of regular
distributions, and that identity pricing is no better than
$\tfrac{1}{2}$-approximate even for the special case of MHR
distributions.

%In the single sample regime, a commonly-used pricing algorithm is
%identity pricing, i.e., let reserve price equal the sampled value.
%\citet{DRY10} showed the identity pricing is $\frac{1}{2}$-approximate
%for regular distributions.
%We will show in \pref{sec:onesamplelb} this is tight in the sense that
%no pricing algorithms using one sample can be better than
%$\frac{1}{2}$-approximation for all regular distributions.
%%While MHR distributions

\begin{comment}
Since MHR distributions are more restrictive than regular ones, one
may expect the identity pricing to give better than
$\frac{1}{2}$-approximation for MHR distributions. However, we show
this is not the case.

\begin{theorem}
\label{thm:mhridentitylb}
The identity pricing algorithm is no better than
$\frac{1}{2}$-approximation for MHR distributions. 
\end{theorem}

\begin{proof}
Consider the uniform distribution over $[1 - \epsilon, 1]$.
The optimal revenue is $1 - \epsilon$, with reserve price $1 - \epsilon$.
The identity pricing algorithm gets revenue $\frac{1}{\epsilon^2}
\int_{1 - \epsilon}^1 v (1 - v) dv = \frac{1}{2} - \frac{1}{3}
\epsilon$. 
So the approximation ratio approaches $\frac{1}{2}$ as $\epsilon$ goes
to zero. 
\end{proof}

We remark that the lower bound still holds if the support must start
from $0$ because we can add a little mass on $[0, 1 - \epsilon]$
without changing the nature of the lower bound.
The same argument applies to the lower bounds in the next section.
\end{comment}

Our next goal is to show that
scaling down the sampled value, i.e., $p(v) = cv$
for some constant $c < 1$, achieves an approximation ratio better than
$\tfrac{1}{2}$ for MHR distributions.

\begin{theorem}
\label{thm:mhrsingleub}
$p(v) = 0.85 v$ is $0.589$-approximate for MHR distributions.
\end{theorem}

The intuition is as follows. We divide the quantile space into two
subsets: those that are larger than the quantile of the optimal
reserve, i.e., $q^*$, and those that are smaller.
\begin{itemize}
\item First, consider those that are larger.
We recall the argument that identity pricing is
$\frac{1}{2}$-approximate: the expected revenue of identity pricing is
the area under the revenue curve; by concavity of the revenue curve,
this is at least half the height and, thus, half the optimal revenue. 
We show that the revenue curve of an MHR distribution is at least as
concave as that of an exponential distribution
(\pref{lem:postpeakrev}).
This implies that identity pricing is strictly better than
$\frac{1}{2}$-approximate for quantiles larger than $q^*$. 
Furthermore, scaling down the price by a factor of $c$ decreases the
revenue by at most a factor of $c$. 
For $c < 1$ close enough to~1,
the expected revenue of this part is still strictly better than one
half of the optimal. 

\item Next, consider quantiles that are smaller than
  $q^*$. \pref{thm:mhridentitylb} suggests that an approximate point
  mass is a 
  worst-case scenario, where the sale probability of identity pricing
  is only $\frac{1}{2}$ on average. 
By scaling down the sampled value by a little, we double the selling
probability w.r.t.\ a point mass without changing the price
by much. 
So the expected revenue of this part is also strictly better than one
half of optimal.
\end{itemize}

In the Appendix, we present the formal argument.

\section{Single Sample Negative Results}
\label{sec:onesamplelb}

First, we note that 
identify pricing is an optimal deterministic strategy
for regular distributions (see Appendix for all proofs).
%\citet{DRY10} showed that the identity pricing is $\frac{1}{2}$-approximate for regular distributions.
%We will show that $\frac{1}{2}$-approximation is the best possible for any algorithms using one sample.

\begin{theorem}
\label{thm:regonesamplelb}
No deterministic $1$-pricing strategy is better than a
$\frac{1}{2}$-approximation for regular distributions.
\end{theorem}

\begin{comment}
\begin{proof}
Distributions with triangle revenue curves (with vertices $(0, 0)$, $(1, 0)$, and $(q^*, R(q^*))$) are commonly considered to be the worst-case regular distributions because they have the least concave revenue curves.
In particular, we consider two such distributions: a point mass at $v$, whose revenue curve is a triangle with $(q^*, R(q^*)) = (1, v)$, and the equal-revenue distribution $F(v) = 1 - \frac{1}{v + 1}$, whose revenue curve is a triangle with $(q^*, R(q^*)) = (1, v)$.

To have non-trivial approximate ratio when the prior is a point mass at $v$, we have $p(v) \le v$.
Then, consider the equal-revenue distribution.
The revenue at price $v$ is $\frac{v}{v + 1}$, which is strictly increasing in $v$.
So any pricing algorithm satisfying $p(v) \le v$ gets revenue less than or equal to that of the identity pricing algorithm $p(v) = v$, which is $2$-approximate for the equal-revenue distribution.
\end{proof}
\end{comment}

%Hu Fu and collaborators recently extended our proof of
%Theorem~\ref{thm:regonesamplelb} to apply also to randomized pricing
%strategies (personal communication, June 2014).

Next, we turn to MHR distributions.
%We will first consider pricing algorithms that are linear in the
%sampled value.
%The reason of focusing on linear pricing algorithms is that we believe the optimal pricing algorithm is linear.
%Suppose we scale the values in a distribution by the same factor.
%The scaling does not change the structure of the algorithm so intuitively the %optimal pricing algorithm should behave identically in the scaled problem.
%This suggests that the pricing algorithm should be scale invariant and, thus, linear.
%Unfortunately, we are unable to formally justify this intuition.
We first present a negative result that holds for
%linear pricing and more generally
every continuously differentiable pricing. Then, we
present a slightly weaker negative result that holds for all
deterministic $1$-pricing algorithms. 

\begin{theorem}
\label{thm:mhronesamplelb1}
No continuously differentiable $1$-pricing strategy is better than a
$0.677$-approximate for MHR distributions.
\end{theorem}

We consider exponential distributions and truncated exponential
distributions, i.e., truncate the all values that are higher than some
threshold $v^*$ in an exponential distribution and replace it with a
uniform distribution over $[v^*, v^* + \alpha]$, where $v^*$ and
$\alpha$ are parameters to be determined later.
We will show that no continuously differentiable pricing functions can
achieve better than $0.677$-approximation in both cases.
We consider these two distributions because of the intuition from the
MHR upper bound analysis: the algorithm need to scale down the sampled
value to handle point mass distributions, but it should not scale down
the value by too much, so that it still gest good revenue for, e.g., the 
exponential distribution. The truncation is an approximation to a
point mass and forces the algorithm to scale down.
%The proof is deferred to \pref{app:proofonesamplelb}.

\begin{theorem}
\label{thm:mhronesamplelb2}
No deterministic $1$-pricing strategy is better than a $\frac{e}{4} \approx 0.68$-approximate for MHR distributions.
\end{theorem}

To prove this more general negative result, we consider a distribution
of exponential distributions, i.e., we first draw a parameter
$\lambda$ from a distribution and then draw the sample $v$ from an
exponential distribution with parameter $\lambda$.
We next solve for the best response pricing strategy w.r.t.\ this distribution of instances, which is the identity pricing.
The theorem then follows from the approximation ratio of identity
pricing for exponential distributions.
%The formal proof is deferred to \pref{app:proofonesamplelb}.

%\input{fewsample}

%\bibliographystyle{acmsmall}
%\bibliographystyle{chicago}
\bibliographystyle{plainnat}
\bibliography{pricing}

\appendix

\section{Missing Proofs from \pref{sec:mhrmanyub}}

\begin{proof}[Proof of \pref{thm:strregmanyub}]
We use the same setup and notation as in the MHR case.
Using the weaker lemmas, we have $q_i, \tilde{q}_i = \Omega(\alpha^{1/(1-\alpha)})$, and $\Delta \ge \Omega(\alpha \delta^2)$.

If $\tilde{q}_1 v_1 \ge \tilde{q}_2 v_2$, then %we have
\[ \frac{\tilde{q}_1}{\tilde{q}_2} \ge \frac{v_2}{v_1} = \frac{R(q_2)}{R(q_1)} \frac{q_1}{q_2} = (1 - \Delta)^{-1} \frac{q_1}{q_2} = \frac{q_1}{q_2} + \Omega(\alpha^{1/(1-\alpha)} \Delta) \enspace.\]
So
\begin{eqnarray*}
\tilde{q}_2 - \tilde{q}_1 & = & \left( 1 - \frac{\tilde{q}_1}{\tilde{q}_2} \right) \tilde{q}_2 ~ \le ~ \left( 1 - \frac{\tilde{q}_1}{\tilde{q}_2} \right) (1 + \epsilon^{3/4}) q_2 \\
& \le & \left( 1 - \frac{q_1}{q_2} - \Omega(\alpha^{1/(1-\alpha)}
\Delta) \right) ( 1 + \epsilon^{3/4} ) q_2 ~ = ~ q_2 - q_1 + \delta
\epsilon^{3/4} - \Omega(\alpha^{1/(1-\alpha)} \Delta).
\end{eqnarray*}
Since $\delta = O(\sqrt{\Delta/\alpha})$ and $\Delta \ge
\frac{\epsilon}{2}$, we have $\delta \epsilon^{3/4} = o(\Delta)$. So
\begin{equation}
%YM EC15
\label{eq:mhrmany1a}
%\label{eq:mhrmany1}
\tilde{q}_2 - \tilde{q}_1 \le q_2 - q_1 - \Omega(\alpha^{1/(1-\alpha)} \Delta) \enspace.
\end{equation}
That is, the number of samples that fall between $q_1$ and $q_2$
differs from its expectation by at least $\Omega(\alpha^{1/(1-\alpha)}
\Delta m)$.
By Chernoff bound, the probability of this event is at most $\exp \big( - \frac{(\alpha^{1/(1-\alpha)} \Delta)^2 m}{\delta} \big)$.
Recall that $\Delta = \Omega(\alpha \delta^2)$, $\Delta \ge
\frac{\epsilon}{2}$, and $m = \Theta(\alpha^{-2/(1-\alpha)-1/2}
\epsilon^{-3/2} \log \epsilon^{-1})$. So this probability is at most
$\exp(-\Omega(\log \epsilon^{-1}) ) = o(\tfrac{1}{m^2})$ with a
suitable choice of parameters.
\end{proof}
%\end{prevproof}

\section{Missing Proofs from \pref{sec:manysamplelb}}

%\begin{prevproof}{Lemma}{lem:dptrick}
\begin{proof}[Proof of \pref{lem:dptrick}]
By the definition of KL divergence, we have
\begin{eqnarray*}
\kl(D_1 \| D_2) + \kl(D_2 \| D_1) & = & \int_\Omega \left[ p_1(\omega) \ln \frac{p_1(\omega)}{p_2(\omega)} + p_2(\omega) \ln \frac{p_2(\omega)}{p_1(\omega)} \right] d \omega \\
& = & \int_\Omega \left[ p_1(\omega) \left( \ln \frac{p_1(\omega)}{p_2(\omega)} + \ln \frac{p_2(\omega)}{p_1(\omega)} \right) + \big( p_2(\omega) - p_1(\omega) \big) \ln \frac{p_2(\omega)}{p_1(\omega)} \right] d \omega \\
& \le & \int_\Omega \bigg[ 0 + | p_2(\omega) - p_1(\omega) | \ln (1 + \epsilon) \bigg] d \omega.
\end{eqnarray*}
The last inequality follows by $(1 + \epsilon)^{-1} \le \frac{f_1(\omega)}{f_2(\omega)} \le (1 + \epsilon)$.
Further, this condition also implies that
\[ | p_2(\omega) - p_1(\omega) | \le \big( (1 + \epsilon) - 1 \big) \min \{ p_1(\omega), p_2(\omega) \} = \epsilon \min \{ p_1(\omega), p_2(\omega) \} \enspace. \]
Thus, we have
\[ \kl(D_1 \| D_2) + \kl(D_2 \| D_1) \le \epsilon \ln(1 + \epsilon) \int_\Omega \min \{ p_1(\omega), p_2(\omega) \} d \omega \le \epsilon \ln(1 + \epsilon) \le \epsilon^2 \enspace. \]
\end{proof}

\subsection{Proof of \pref{thm:mhrmanylb}}

Recall that the revenue curves of $D_1$ and $D_2$ are summarized in \pref{fig:mhrmanylb}.

\begin{lemma}
\label{lem:klmhr}
\[ \kl(D_1 \| D_2) + \kl(D_2 \| D_1) = O(\epsilon_0^{3/2}) \enspace.\]%\le \frac{4 \epsilon^2}{1 - \sqrt{\epsilon}} + \frac{4 \epsilon^{3/2}}{(1 - 2 \sqrt{\epsilon})^2} \]
\end{lemma}

\begin{proof}%[Sketch]
By our choice of $D_1$ and $D_2$, $D_1$ and $D_2$ differs by $1 +
\frac{2 \epsilon_0}{1 - \sqrt{\epsilon_0}}$ for $v \in [1 +
  \sqrt{\epsilon_0}, 2]$, and differs by $1 - 2 \sqrt{\epsilon_0}$ for
$v \in [1, 1 + \sqrt{\epsilon_0}]$.
The lemma follows from \pref{lem:dptrickmhr}.
\end{proof}

Let $q^*_i$ be the revenue optimal quantile of $R_i$, and let $R^*_i = R_i(q^*_i)$ be the optimal revenue.
\begin{lemma}
$q^*_1 = 1$, and $R^*_1 = 1$.
\end{lemma}

\begin{proof}
$R_1(q) = q (2 - q)$ is a quadratic curve that peaks at $q = 1$ with $R_1(1) = 1$.
\end{proof}

\begin{lemma}
$q^*_2 = 1 - \sqrt{\epsilon_0} + 2\epsilon_0$, and $R^*_2 = 1 + \epsilon_0 + 2 \epsilon^{3/2}$.
\end{lemma}

\begin{proof}
For $1 \le v \le 1 + \sqrt{\epsilon_0}$ and $1 - \sqrt{\epsilon_0} + 2\epsilon_0 \le q(v) \le 1$, $R_2$ is identical to the revenue curve of the uniform distribution over $[1, 1 + \frac{1}{1 - 2 \sqrt{\epsilon_0}}]$, i.e., $q \big(\frac{2 - 2 \sqrt{\epsilon_0}}{1 - 2 \sqrt{\epsilon_0}} - \frac{1}{1 - 2 \sqrt{\epsilon_0}} q \big)$, which is maximized at $q = 1 - \sqrt{\epsilon_0}$.

For $1 + \sqrt{\epsilon_0} \le v \le 2$ and $0 \le q(v) \le 1 - \sqrt{\epsilon_0} + 2 \epsilon_0$, $R_2$ is identical to the revenue curve of the uniform distribution over $\big[ 1 + \frac{2\epsilon_0}{1 - \sqrt{\epsilon_0} + 2\epsilon_0}, 2 \big]$, i.e., $q \big( 2 - \frac{1 - \sqrt{\epsilon_0}}{1 - \sqrt{\epsilon_0} + 2 \epsilon_0} q \big)$, which is maximized at $q = 1$ in interval $[0, 1]$.

$R_2$ is maximized at the intersection of the two parts, where $q = 1
- \sqrt{\epsilon_0} + 2\epsilon_0$, and $R_2(1 - \sqrt{\epsilon_0} +
2\epsilon_0) = 1 + \epsilon_0 + 2 \epsilon^{3/2}$.
\end{proof}

%Let $q^*_1 = 1$ and $q^*_2 = 1 - \sqrt{\epsilon} + \epsilon$ be the optimal quantiles.
%Let $R^*_1 = R_1(q^*_1)$ and $R^*_2 = R_2(q^*_2)$ be the maximal revenue.
The following lemmas follow from our construction of $D_1$ and $D_2$.

\begin{lemma}
If $R_1(q) \ge (1 - 3\epsilon) R^*_1$, then $q \ge 1 - \sqrt{3\epsilon}$ and $v(q) \le \frac{1 - 3\epsilon}{1 - \sqrt{3\epsilon}} R^*_1$.
\end{lemma}

\begin{lemma}
If $R_2(q) \ge (1 - 3\epsilon) R^*_2$, then $q \le 1 - \sqrt{\epsilon_0} + 2 \epsilon_0 + \sqrt{3\epsilon}$ and $v(q) \ge \frac{1 - 3\epsilon}{1 - \sqrt{\epsilon_0} + 2 \epsilon_0 + \sqrt{3\epsilon}} R^*_2$.
\end{lemma}

%\begin{proof}[Sketch]
%Note that $R_1$ and $R_2$ are (piecewise) quadratic functions over $q$.
%So $R_i(q) - R_i(q^*_i) = \Omega(q - q^*_i)^2$.
%The lemmas then follow.
%\end{proof}

Therefore, when $\epsilon_0 = c \epsilon$ for sufficiently large constant $c$, we have $1 - \sqrt{3\epsilon} > 1 - \sqrt{\epsilon_0} + 2 \epsilon_0 + \sqrt{3\epsilon}$.
Further note that $R^*_2 > R^*_1$, so the $(1 - 3\epsilon)$-optimal price sets of $D_1$ and $D_2$ are disjoint.
\pref{thm:mhrmanylb} then follows from \pref{thm:reduction} and \pref{lem:klmhr}.

\section{Missing Proofs from \pref{sec:onesampleub}}

\subsection{Lower Bound for Identity Pricing}

Since MHR distributions are more restrictive than regular ones, one
might expect identity pricing to give better than
$\frac{1}{2}$-approximation for then.
This is not the case.

\begin{theorem}
\label{thm:mhridentitylb}
The identity pricing algorithm is no better than $\frac{1}{2}$-approximation for MHR distributions.
\end{theorem}

\begin{proof}
Consider the uniform distribution over $[1 - \epsilon, 1]$.
The optimal revenue is $1 - \epsilon$, with reserve price $1 - \epsilon$.
The identity pricing algorithm gets revenue $\frac{1}{\epsilon^2}
\int_{1 - \epsilon}^1 v (1 - v) dv = \frac{1}{2} - \frac{1}{3}
\epsilon$.
The approximation ratio approaches $\frac{1}{2}$ as $\epsilon$ goes
to zero.
\end{proof}

We remark that the lower bound still holds if the support must start
from $0$ because we can add a little mass on $[0, 1 - \epsilon]$
without changing the nature of the lower bound.
The same trick applies to the lower bounds in
Section~\ref{sec:onesamplelb}.

\subsection{Proof of \pref{thm:mhrsingleub}}

Let $\hat{R}(q)$ be the expected revenue of reserve price $p(v(q))) = c v(q)$ w.r.t.\ revenue curve $R$.
Define $\hat{R}^\mhr(q)$ similarly w.r.t.\ the exponential distribution $D^\mhr$.
For ease of presentation, we will w.l.o.g.\ scale the values so that $v(q^*) = v^\mhr(q^*)$ and, thus, $R(q^*) = R^\mhr(q^*)$.

\subsection*{Technical Lemmas about MHR Distributions}

Given the revenue at quantile $q^*$, $R^\mhr$ minimizes the revenue at any quantile larger than $q^*$ among all MHR distributions.  %(Hence, minimizes the revenue at any price below the monopoly price.)

\begin{lemma}
\label{lem:postpeakrev}
For any $1 \ge q \ge q^*$, then $R(q) \ge R^\mhr(q)$.
\end{lemma}

\begin{proof}
Suppose not.
Let $q_0$ be a quantile such that $R(q_0) < R^\mhr(q_0)$.
%Let
%\[ q_1 = \inf_{q \in [q^*, q_0]} \bigg\{ R(q) < R^\mhr(q) \bigg\} \]
Since $R(q^*) = R^\mhr(q^*)$ and that $R$ and $R^\mhr$ are continuous, there exists $q_1 \in [q^*, q_0]$ such that $R(q_1) = R^\mhr(q_1)$ and $R'(q_1) \le (R^\mhr)'(q_1)$.
Similarly, since $R(1) \ge 0 = R^\mhr(1)$, there exists $q_2 \in [q_0, 1]$ such that $R(q_2) < R^\mhr(q_2)$ and $R'(q_2) > (R^\mhr)'(q_2)$.

Recall that $R'(q_1) = \phi(v(q_1))$ and $R'(q_2) = \phi(v(q_2))$.
By the MHR assumption,
\begin{equation}
\label{eq:postpeakref1}
R'(q_1) - R'(q_2) = \phi(v(q_1)) - \phi(v(q_2)) \ge v(q_1) - v(q_2) =
\frac{R(q_1)}{q_1} - \frac{R(q_2)}{q_2}.
\end{equation}

By the definition of $R^\mhr$, the above relation holds with equality for $R^\mhr$:
\begin{equation}
\label{eq:postpeakref2}
(R^\mhr)'(q_1) - (R^\mhr)'(q_2) = \frac{R^\mhr(q_1)}{q_1} - \frac{R^\mhr(q_2)}{q_2}
\end{equation}

By \pref{eq:postpeakref1} minus \pref{eq:postpeakref2} and that $R(q_1) = R^\mhr(q_1)$, we have
\[ R'(q_1) - R'(q_2) - (R^\mhr)'(q_1) + (R^\mhr)'(q_2) \ge \frac{R^\mhr(q_2)}{q_2} - \frac{R(q_2)}{q_2} > 0 \enspace, \]
contradicting $R'(q_1) \le (R^\mhr)'(q_1)$ and $R'(q_2) > (R^\mhr)'(q_2)$.
\end{proof}

%\begin{remark}
The above lemma does not apply to quantiles smaller than $q^*$.
Consider a point mass, or a uniform distribution over $[1 - \epsilon, 1]$ for sufficiently small $\epsilon$ for a counter example.
The problem with applying it to quantiles smaller than $q^*$ is that $q_1$ might be $0$.
So we have zero divided by zero issue in \pref{eq:postpeakref1} and \pref{eq:postpeakref2} and, thus, $R(q_1) / q_1$ may not equal $R^\mhr(q_1) / q_1$ when we subtract \pref{eq:postpeakref2} from \pref{eq:postpeakref1}.
%\end{remark}

As a corollary of \pref{lem:postpeakrev}, $R^\mhr$ minimizes the value at each quantile larger than $q^*$ (\pref{lem:postpeakval}), and it maximizes the quantile at each value smaller than $v(q^*)$ (\pref{lem:postpeakquan}).

\begin{lemma}
\label{lem:postpeakval}
For any $1 \ge q \ge q^*$, then $v(q) \ge v^\mhr(q)$.
\end{lemma}

\begin{lemma}
\label{lem:postpeakquan}
For any $0 \le v \le v(q^*)$, then $q(v) \ge q^\mhr(v)$.
\end{lemma}

\subsection*{Expected Revenue of Large Quantiles}

We show that for quantiles between $q^*$ and $1$, $R^\mhr$ is indeed the worst-case scenario.

\begin{lemma}
\label{lem:mhrpostpeaklocal}
For any $1 \ge q \ge q^*$, $\hat{R}(q) \ge \hat{R}^\mhr(q)$.
\end{lemma}

\begin{proof}
We abuse notation and let $R(v) = R(q(v))$.
By \pref{lem:postpeakval}, $v(q) \ge v^\mhr(q)$.
So we have $c v(q) \ge c v^\mhr(q)$. Since $R(v)$ is non-decreasing when $v > v(q^*)$ and $v(q^*) > c v(q) \ge c v^\mhr(q)$, we have $\hat{R}(q) = R \big( c v(q) \big) \ge R \big( cv^\mhr(q) \big)$.

Further, by \pref{lem:postpeakquan}, $R^\mhr$ minimizes the sale probability at any price $v \le v(q^*)$ and, thus, minimizes the revenue at price $v$. So we have $R \big( cv^\mhr(q) \big) \ge R^\mhr \big( cv^\mhr(q) \big) = \hat{R}^\mhr(q)$ and the lemma follows.
\end{proof}

As direct corollary of \pref{lem:mhrpostpeaklocal}, we can lower bound the expected revenue of quantiles between $q^*$ and $1$ by the expected revenue of the worst-case distribution $R^\mhr$.

\begin{lemma}
\label{lem:mhrpostpeak}
$\int^1_{q^*} \hat{R}(q) dq \ge \int^1_{q^*} \hat{R}^\mhr(q) dq$.
\end{lemma}

\subsection*{Expected Revenue of Small Quantiles}
%\label{sec:mhrsmallquantile}

Next, we consider quantiles between $0$ and $q^*$.
Let $q_0$ be that $c v(q_0) = v(q^*)$, i.e., the reserve price is
smaller than $v(q^*)$ if and only if the sample has quantile larger
than $q_0$. We will first lower bound the expected revenue of
quantiles smaller than $q_0$, and then handle the other quantiles.

\begin{lemma}
\label{lem:mhrsmallprepeak1}
For any $0 \le q' \le q_0$, $\int^{q'}_0 \hat{R}(q) dq \ge \tfrac{1 + \sqrt{1-c}}{2} q' R(q')$.
\end{lemma}

\begin{proof}
Let $c_0 = \frac{c}{2}$. For any $i \ge 0$, let $c_{i+1} = 1 - \tfrac{c}{4 c_i}$. We will inductively show that $\int^{q'}_0 \hat{R}(q) dq \ge c_i q' R(q')$, and then prove that $c_i$ converges to $\tfrac{1 + \sqrt{1-c}}{2}$.

\paragraph{Base Case}
By concavity of the revenue curve, $\int^{q'}_0 R(q) dq \ge \frac{1}{2} R(q')$.
Further, lower reserve prices has larger sale probability.
So $\hat{R}(q) \ge c R(q)$ and the base case follows.

\paragraph{Inductive Step}
%Next, we explain the inductive step.
Let $q_1$ be that $c v(q_1) = v(q')$, i.e., the reserve price is
smaller than $v(q')$ iff the sample has quantile larger than $q_1$. We
have
\[ R(q_1) = v(q_1) q_1 = \frac{v(q') q_1}{c} = \frac{q_1}{c q'} q' v(q') = \frac{q_1}{c q'} R(q'). \]
For the interval from $0$ to $q_1$, by inductive hypothesis we have
\[ \int^{q_1}_0 \hat{R}(q) dq \ge c_i q_1 R(q_1) = \frac{c_i q_1^2}{c q'} R(q') \]
Next, consider the interval from $q_1$ to $q'$. For any $q_1 \le q \le q'$, by that $q' \le q_0$ and our choice of $q_0$, we have $v(q^*) < c v(q) < c v(q_1) = v(q')$. So $\hat{R}(q) \ge R(q')$. Thus,
\[ \int^{q'}_{q_1} \hat{R}(q) dq \ge (q' - q_1) R(q'). \]
Putting together, we have
\[ \int^{q'}_0 \hat{R}(q) dq \ge \bigg( \frac{c_i}{c}
\bigg(\frac{q_1}{q'}\bigg)^2 + \bigg( 1 - \frac{q_1}{q'} \bigg) \bigg)
q' R(q'). \]
Minimizing the RHS over $0 \le q_1 \le q'$, we have
\[ \int^{q'}_0 \hat{R}(q) dq \ge \bigg( 1 - \frac{c}{4 c_i} \bigg) q' R(q') = c_{i+1} q' R(q'). \]

\paragraph{Convergence of $c_i$}
There is only one stable stationary point, $\frac{1 + \sqrt{1-c}}{2}$,
for the recursion $c_{i+1} = 1 - \frac{c}{4 c_i}$.
The other, non-stable, stationary point is $\frac{1 -
  \sqrt{1-c}}{2}$. Note that for any $0 < c < 1$, $c_0 = \frac{c}{2} >
\frac{1 - \sqrt{1-c}}{2}$.
So $c_i$ converges to $\frac{1 + \sqrt{1-c}}{2}$ as $i$ goes to
infinity.
%The lemma then follows.
\end{proof}

%So as a corollary of \pref{lem:mhrsmallprepeak1}, we have

\begin{lemma}
\label{lem:prepeak}
If $c = 0.85$, then $\int^{q^*}_0 \hat{R}(q) dq \ge 0.656 q^* R(q^*)$.
\end{lemma}

\begin{proof}%[Sketch]
Recall that $c v(q_0) = v(q^*)$. So
\[ R(q_0) = q_0 v(q_0) = \frac{q_0}{c q^*} q^* v(q^*) = \frac{q_0}{c
  q^*} R(q^*). \]
Plugging $c = 0.85$ and $\frac{1 + \sqrt{1-c}}{2} \ge 0.693$ into \pref{lem:mhrsmallprepeak1}, we have
\begin{equation}
\label{eq:prepeak1}
\int^{q_0}_0 \hat{R}(q) dq \ge 0.693 q_0 R(q_0) = 0.693 \frac{1}{c} \bigg( \frac{q_0}{q^*} \bigg)^2 q^* R(q^*) \ge 0.815 \bigg( \frac{q_0}{q^*} \bigg)^2 q^* R(q^*).
\end{equation}
On the other hand, for every $q_0 \le q \le q^*$, by concavity of the
revenue curve, we have
\[ q v(q) \ge \frac{q - q_0}{q^* - q_0} q^* v(q^*) + \frac{q^* - q}{q^* - q_0} q_0 v(q_0). \]
Thus,
\[ v(q) \ge \frac{q^* v(q^*) - q_0 v(q_0)}{q^* - q_0} + \frac{1}{q} \frac{q_0 q^*}{q^* - q_0} \big( v(q_0) - v(q^*) \big). \]
Further, by our choice of $q_0$, the quantile of $c v(q)$ is at least $q^*$. So we have
\[ \hat{R}(q) \ge c v(q) q^* \ge \bigg( \frac{q^* v(q^*) - q_0
  v(q_0)}{q^* - q_0} + \frac{1}{q} \frac{q_0 q^*}{q^* - q_0} \big(
v(q_0) - v(q^*) \big) \bigg) c q^*.\]
Let $x = \frac{q_0}{q^*}$. Plugging in $c v(q_0) = v(q^*)$, $R(q_0) = \frac{q_0}{c q^*} R(q^*)$, and $c = 0.85$, we have
\[ \bigg( \frac{0.85 - x}{q^* - q_0} + \frac{1}{q} \frac{(1 - 0.85) x}{1 - x} \bigg) q^* R(q^*). \]
Integrating over $q$ from $q_0$ to $q^*$, we have
\begin{equation}
\label{eq:prepeak2}
\int^{q^*}_{q_0} \hat{R}(q) dq \ge \bigg( (0.85 - x) + \ln
\bigg(\frac{1}{x}\bigg) \frac{0.15 x}{1 - x} \bigg) q^* R(q^*).
\end{equation}
Summing up \pref{eq:prepeak1} and \pref{eq:prepeak2} gives
\[ \int^{q^*}_0 \hat{R}(q) dq \ge \bigg( (0.85 - x) + \ln
\bigg(\frac{1}{x}\bigg) \frac{0.15 x}{1 - x} + 0.815 x^2 \bigg) q^*
R(q^*). \]
%Numerically minimizing over $0 \le x \le 1$ we show the lemma.

%[[YM: doing it more explicitly]]

We would like to minimize $f(x)$ for $x \in [0,1]$, where
\[
f(x)= (0.85 - x) + \ln \bigg(\frac{1}{x}\bigg) \frac{0.15 x}{1 - x}
+ 0.815 x^2.
\]
Taking the derivative we have
\[
f^\prime(x)= -1-\frac{0.15}{1-x}-\frac{0.15 \ln(x)}{(1-x)^2}+1.63x.
\]
This function has two roots in $x\in[0, 1]$, at $x\in[ 0.546,0.547]$
($f^\prime(0.546)<-3*10^{-5}$ and $f^\prime(0.547)>10^{-3}$) and at
$x\in[ 4.7\cdot 10^{-4},4.8\cdot 10^{-4}]$
($f^\prime(4.7\cdot10{-4})>1.3\cdot 10^{-3}$ and
$f^\prime(4.8\cdot10^{-4})< 1.9\cdot 10^{-3}$). Testing the second
derivative
\[
f^{\prime\prime}(x)= -\frac{0.15}{(1-x)^2}-\frac{0.3
\ln(x)}{(1-x)^3}-\frac{0.15}{x(1-x)^2}+1.63
\]
we have that $x\approx 0.546$ is a minimum point
($f^{\prime\prime}(0.546)\approx 1.5$)  and $x\approx 4.7*10^{-4}$
is a local maximum ($f^{\prime\prime}(4.7*10^{-4})\approx -315$).
Therefore, the only remaining point we need to test is $x=0$ (the
end of the interval $[0,1]$) and we have $\lim_{x\rightarrow 0} f(x)
=0.85$.
\end{proof}

\subsection*{Proof of \pref{thm:mhrsingleub}}

Plugging in the exponential distribution in \pref{lem:mhrpostpeak}, we
obtain
\[\int^1_{q^*} \hat{R}(q) dq \ge \frac{c}{(c+1)^2} \frac{1 -
  (q^*)^{c+1} + (q^*)^{c+1} \ln (q^*)^{c+1}}{- (q^*) \ln (q^*)}
R(q^*).\]
Combining this with \pref{lem:prepeak} yields
\[\int^1_0 \hat{R}(q) dq \ge \bigg(0.656 q^* + \frac{c}{(c+1)^2}
\frac{1 - (q^*)^{c+1} + (q^*)^{c+1} \ln (q^*)^{c+1}}{- (q^*) \ln
  (q^*)} \bigg) R(q^*).\]

%[[YM: More formal derivation]]

We would like to lower bound $f(x)$, where
\[
f(x)= 0.656x+\frac{c}{(c+1)^2}\bigg( -x^{-1}\ln^{-1}(x) + x^c
\ln^{-1}(x) -(c+1)x^{c}\bigg).
\]
The derivative is
\[
f^\prime(x)= 0.656+\frac{c}{(c+1)^2}\bigg( x^{-2}\ln^{-1}(x) +
x^{-2}\ln^{-2}(x) + cx^{c-1} \ln^{-1}(x)-x^{c-1}
\ln^{-2}(x)-c(c+1)x^{c}\bigg).
\]
For $c=0.85$, we have a root at some
$x\in[0.544,0.545]$ as $f^\prime(0.544)<-1.7\cdot 10^{-4}$ and
$f^\prime(0.545)>9\cdot 10^{-4}$. By testing the second derivative, this is a local minimum and the only root in the interval.
The minimal value is at least $0.589$ and the lemma follows.

\section{Missing Proofs in \ref{sec:onesamplelb}}
\label{app:proofonesamplelb}

%\begin{prevproof}{Theorem}{thm:regonesamplelb}
\begin{proof}[Proof of \pref{thm:regonesamplelb}]
Distributions with triangle revenue curves (with vertices $(0, 0)$,
$(1, 0)$, and $(q^*, R(q^*))$) are commonly considered to be the
worst-case regular distributions because they have the least concave
revenue curves. 
In particular, we consider two such distributions: a point mass at
$v$, whose revenue curve is a triangle with $(q^*, R(q^*)) = (1, v)$,
and the distribution $F(v) = 1 - \frac{1}{v + 1}$, whose
revenue curve is a triangle with $(q^*, R(q^*)) = (0, 1)$. 

To achieve a non-trivial approximation ratio when the 
distribution is a point mass at $v$, $p(v) \le v$ must hold.
Then, consider the second distribution.  The revenue at price $v$ is
$\frac{v}{v + 1}$, which is strictly increasing in $v$. 
So every deterministic pricing algorithm that satisfies $p(v) \le v$
gets revenue less than or equal to that of the identity pricing
algorithm $p(v) = v$, which is $2$-approximate for the second
distribution.
\end{proof}
%\end{prevproof}

\begin{proof}[Proof of \pref{thm:mhronesamplelb1}]
Let us first assume the pricing algorithm is linear, i.e., $b(v) =
cv$. 
As in the proof of \pref{thm:regonesamplelb}, the algorithm has
a finite approximation ratio only if $c \le 1$.

We consider exponential distributions and truncated exponential
distributions, i.e., truncate the all values that are higher than some
threshold $v^*$ in an exponential distribution and replace it with a
uniform distribution over $[v^*, v^* + \alpha]$, where $v^*$ and
$\alpha$ are parameters to be determined later. 
We show that no scaling parameter $c$ can achieve better than
$0.67$-approximation in both cases.  

We consider these two distributions because of the intuition from the
MHR upper bound analysis: the algorithm needs to scale down the sampled
value to handle point mass case, but it should not scale down the
value by too much in order to get good revenue for, e.g., the
exponential distribution. The truncation is an approximation to a
point mass and forces the algorithm to scale down.%, while the
                                %untruncated part and the exponential
                                %distribution prevents the algorithm
                                %from scaling down by too much. 

\paragraph{Exponential Distribution}

Consider the approximation ratio w.r.t.\ the exponential distribution
$D^\mhr$.  
Given a sample with quantile $q$, its value is $v^\mhr(q) = - \ln q$.
So $b(v(q)) = - c \ln q = - \ln q^c$ and the sale probability is $q^c$.
The expected revenue of $b(v)$ is 
\[ \int^1_0 - q^c \ln q^c dq = \frac{c}{(c+1)^2} \enspace. \]
Recall that $R^\mhr(q) = - q \ln q$, which is maximized at $q = \frac{1}{e}$ with optimal revenue $\frac{1}{e}$. 
So the approximation ratio w.r.t.\ $D^\mhr$ is $\frac{e c}{(c + 1)^2}$. 
Note that this immediately gives a lower bound of $\frac{e}{4} \approx 0.68$.
If $c \le 0.878$, then $\frac{e c}{(c + 1)^2} < 0.677$. 
For now on, we assume that $c \ge 0.878$. 

\paragraph{Truncated Exponential Distribution}

Let $q^* = 0.43$ and consider an exponential distribution such that $v(q^*) = 1$.
Truncate the exponential distribution at $q^*$ with a uniform distribution over $[1, 1 + \alpha]$ with $\alpha = 0.74$.
Hence, for $q \in [q^*, 1]$, $v(q) = \frac{\ln q}{\ln q^*}$; for $q \in [0, q^*]$, $v(q) = 1 + \alpha \big( 1 - \frac{q}{q^*} \big)$. 
It is easy to check that the revenue is maximized at $q = q^*$ with maximal revenue $q^*$.% for this distribution.

Next, we upper bound the expected revenue of $b(v)$.
Consider first the contribution from quantiles $q \in [\frac{1}{e},
  1]$. The analysis is similar to that of $D^\mhr$, except that the
values are scaled up by $- \frac{1}{\ln q^*}$. So this part
contributes 
\begin{eqnarray*}
- \frac{1}{\ln q^*} \int^1_{q^*} - q^c \ln q^c dq & = & - \frac{1}{\ln q^*} \frac{c}{(c+1)^2} \bigg(1 - (q^*)^{c+1} + (q^*)^{c+1} \ln (q^*)^{c+1} \bigg) \\
& \le & 2.7556 \frac{c}{(c+1)^2} \bigg(1 - (q^*)^{c+1} + (q^*)^{c+1} \ln (q^*)^{c+1} \bigg) q^* \enspace.
\end{eqnarray*}

Now consider the quantiles that are smaller than $q^*$ and have values
at most $\frac{1}{c}$, i.e., the corresponding reserve price is
smaller than $1$.  
These are quantiles $q \in [\frac{c \alpha + c - 1}{c \alpha} q^*, q^*]$.
We upper bound the expected revenue by the optimal revenue $q^*$ when the sampled quantile is in this interval. So this part contributes at most
\[ \left( q^* - \frac{c \alpha + c - 1}{c \alpha} q^* \right) q^* =
\left( \frac{1 - c}{c \alpha} \right) (q^*)^2 \le 0.5811 \left(
\frac{1}{c} - 1 \right) q^* \enspace. \] 

Finally, consider quantiles $q \in [0, \frac{c \alpha + c - 1}{c \alpha} q^*]$. 
The corresponding value is $v(q) = 1 + \alpha \big( 1 - \frac{q}{q^*} \big)$ and the reserve price is $b(v) = c \big( 1 + \alpha \big( 1 - \frac{q}{q^*} \big) \big)$. The sale probability at this price is 
$\frac{1 + \alpha - b(v)}{\alpha} q^*$. % \enspace. \]% = \frac{1 + \alpha - c (1 + \alpha (1 - eq))}{e \alpha} \enspace. \]
So this part contributes
\[ \int^{\frac{c \alpha + c - 1}{c \alpha} q^*}_0 b(v) \frac{1 +
  \alpha - b(v)}{\alpha} q^* dq. \] 
Note that $db(v) = - \frac{c \alpha}{q^*} dq$, $b = 1$ when $q = \frac{c \alpha + c - 1}{c \alpha} q^*$, and $b = c (1 + \alpha)$ when $q = 0$. So the above equals
\begin{eqnarray*}
\frac{(q^*)^2}{c \alpha^2} \int^{c (1 + \alpha)}_0 b(v) (1 + \alpha - b(v)) db(v) & = & \frac{(q^*)^2}{c \alpha^2} \left( \frac{1 + \alpha}{2} \left( \big( c (1 + \alpha) \big)^2 - 1 \right) - \frac{1}{3} \left( \big( c (1 + \alpha) \big)^3 - 1 \right) \right) \\
& = & \frac{(q^*)^2}{c \alpha^2} \left( \left(\frac{c^2}{2} - \frac{c^3}{3}\right) (1 + \alpha)^3 - \frac{1 + \alpha}{2} + \frac{1}{3} \right) \\
& \le & \left( -1.3789 c^2 + 2.0683 c - 0.4215 \frac{1}{c} \right)
q^*. 
\end{eqnarray*}

Putting everything together and dividing the expected revenue by the optimal
revenue $q^*$, the approximation ratio is at most  
\[ 2.7556 \frac{c}{(c+1)^2} \bigg( 1 - (q^*)^{c+1} + (q^*)^{c+1} \ln (q^*)^{c+1} \bigg) + 0.5811 \left( \frac{1}{c} - 1 \right)+ \left( -1.3768 c^2 + 2.2683 c - 0.4215 \frac{1}{c} \right). \]
%(Zhiyi: I will fill in some more derivation later if have time)
We numerically maximize the above function over $c \in [0.878, 1]$.
It is decreasing in the interval $[0.878, 1]$ and takes value about $0.6762$ at $c = 0.878$. So the lemma follows for linear pricing functions.

\paragraph{Continuously Differentiable Pricing Functions}

Next, we explain how to reduce the case of a continuously
differentiable pricing function to the linear pricing function case.  
Since $b(v)$ is continuously differentiable, for any $\delta > 0$,
there exists $\epsilon > 0$ such that for any $v \in [0, \epsilon]$,
$|b'(v) - b'(0)| < \delta$, i.e., in this neighborhood of $0$, $b(v)$
behaves like a linear function with slope approximately $b'(0)$ (up to
error $\delta$). 
So we can handle them like the linear case.% and let $\delta$ goes to zero. 

Formally, if $b'(0) < \frac{1}{2}$, then $b(v) \le (\frac{1}{2} + \delta) v$ for $v \le \epsilon$.
So its approximation ratio w.r.t.\ a point mass at $v$ is at most $\frac{1}{2} + \delta < 0.677$ for sufficiently small $\delta$.

Next, assume $b'(0) > \frac{1}{2}$.
Let us scale down the values in the exponential distribution and the truncated exponential distribution from the linear case such that all values are less than $\epsilon$.\footnote{Note that the support of the exponential distribution spans all non-negative real numbers. So instead of scaling to make all values smaller than $\epsilon$, we will make sure $1 - 10^{-5}$ fraction of the values are smaller than $\epsilon$; the remaining values can change the approximation ratio by at most $10^{-5}$.}
For any sampled value $v$, the expected revenue is $b(v) q(b(v))$ where $(f'(0) - \delta) v < b(v) < (f'(0) + \delta) v$. So 
\[ b(v) q(b(v)) \le (f'(0) + \delta) v q((f'(0) - \delta) v) = \frac{f'(0) + \delta}{f'(0) - \delta} (f'(0) - \delta) v q((f'(0) - \delta) v)\enspace, \] 
which is at most $\frac{f'(0) + \delta}{f'(0) - \delta}$ times larger
than the revenue of linear pricing function $(f'(0) - \delta) v$. 
The lemma then follows by $f'(0) \ge \frac{1}{2}$ and letting $\delta$
goes to zero. 
\end{proof}

\begin{proof}[Proof of \pref{thm:mhronesamplelb2}]
By previous arguments, it suffices to consider only (deterministic)
pricing functions with $p(v) \le v$ for all $v$.
We consider a distribution over value distributions. 
Draw $\lambda$ from an exponential distribution with parameter $\gamma$, i.e., the density of $\lambda$ is $\gamma e^{- \gamma \lambda}$, where $\gamma$ is a parameter to be determined later.
Let the value distribution to be an exponential distribution with parameter $\lambda$, i.e., the density of $v$ is $\lambda e^{- \lambda v}$.

We first compute the best response pricing algorithm $p(v)$ subject to
$p(v) \le v$ for this case.
The expected revenue of $p(v)$ is
\begin{eqnarray*}
R & = & \int_0^\infty \left[\int_0^\infty \lambda e^{-\lambda v} p(v)e^{-\lambda p(v)} dv\right] \gamma e^{-\gamma \lambda} d\lambda\\
& = & \gamma\int_0^\infty \left[\int_0^\infty \lambda e^{-\lambda (v+p(v)+\gamma)} p(v) d\lambda\right] dv \\
& = & \gamma \int_0^\infty \left[\int_0^\infty \lambda e^{-\lambda (v+p(v)+\gamma)} p(v) d\lambda\right] dv \\
& = & \gamma \int_0^\infty \left[\frac{p(v)}{v+p(v)+\gamma}\int_0^\infty (v+p(v)+\gamma)\lambda e^{-\lambda (v+p(v)+\gamma)}  d\lambda\right] dv \\
& = & \gamma \int_0^\infty \frac{p(v)}{(v+p(v)+\gamma)^2} dv.
\end{eqnarray*}

Note that $\frac{p(v)}{(v+p(v)+\gamma)^2}$ is maximized at $p(v) = v$ for $p(v) \le v$.
So the best response is $p(v) = v$.
Given any $\lambda$, the optimal revenue is $\frac{1}{e \lambda}$, and
the expected revenue of $p(v) = v$ is $\frac{1}{4 \lambda}$. 
So the approximation ratio is at most $\frac{e}{4} \approx 0.68$, as desired.
\end{proof}

\end{document}